\title{Convex scalarizations of the mean-variance-skewness-kurtosis problem in portfolio selection}

\author{Andries Steenkamp
\thanks{Centrum Wiskunde \& Informatica (CWI), Amsterdam. \url{andries.steenkamp@cwi.nl}
\newline
This work is supported by the European Union's Framework Programme for Research and Innovation Horizon
2020 under the Marie Skłodowska-Curie Actions Grant Agreement No. 813211  (POEMA).
}}

\date{\today}

\documentclass[11pt]{article}

\usepackage{fullpage}
\usepackage[utf8]{inputenc}
\usepackage[english]{babel}
\usepackage{amsmath}
\usepackage{amsthm}
\usepackage{amssymb}
\usepackage{blkarray}
\usepackage{mathtools}
\usepackage{mathrsfs}
\usepackage{xcolor}

\usepackage{placeins}

\usepackage[bookmarksnumbered,hypertexnames=false,colorlinks=false,linkcolor=blue,urlcolor=blue,citecolor=blue,anchorcolor=green,breaklinks=true,pdfusetitle]{hyperref}
\usepackage[capitalize]{cleveref}

\usepackage{graphicx}
\usepackage[lofdepth,lotdepth]{subfig}

\date{\today}

\newtheorem{theorem}{Theorem}

\newtheorem{lemma}[theorem]{Lemma}
\newtheorem{proposition}[theorem]{Proposition}
\newtheorem{corollary}[theorem]{Corollary}

\newtheorem{observation}[theorem]{Observation}

\newcommand{\bb}[1]{\mathbb{#1}} 
\newcommand{\R}{\bb{R}}

\newcommand{\inp}[1]{\langle #1 \rangle}
\newcommand{\qo}[1]{ \overline{#1} }         
\newcommand{\qu}[1]{ \underline{#1} }         
\newcommand{\argmin}{\mathrm{argmin}}
\newcommand{\supp}{\mathrm{supp}}
\newcommand{\corr}{\mathrm{corr}}
\newcommand{\wtl}{\widetilde }
\newcommand{\wh}{\widehat }

\newcommand{\proj}{\mathbf{\mathrm{Proj}}}

\newcommand{\vol}{\mathrm{vol}}

\newcommand{\qua}{\Phi_{\lambda}(R,w)}
\newcommand{\quay}{\Psi_{\lambda}(y)}

\newcommand{\Rtri}{\qo{R}_{\Delta}}
\newcommand{\Rtriu}{\qu{R}_{\Delta}}

\newcommand{\Rsqr}{\qo{R}_{\square}}
\newcommand{\Rsqru}{\qu{R}_{\square}}

\newcommand{\Fd}[1]{F_{#1,\Delta}^{[40]}}
\newcommand{\Fid}{\Fd{i}}
\newcommand{\Fstar}[1]{F_{\Delta}^{[40],#1}}
\newcommand{\Fdk}[1]{F_{#1,\Delta,k}^{[40]}}
\newcommand{\Fdfive}[1]{F_{#1,\Delta,5}^{[40]}}
\newcommand{\Fkstar}[1]{F_{\Delta,k}^{[40],#1}}
\newcommand{\Ffivestar}[1]{F_{\Delta,5}^{[40],#1}}

\newcommand{\Fs}[1]{F_{#1,\square}^{[40]}}
\newcommand{\Fis}{\Fs{i}}
\newcommand{\Fstars}[1]{F_{\square}^{[40],#1}}

\newcommand{\Fsk}[1]{F_{#1,\square,k}^{[40]}}
\newcommand{\Fsfive}[1]{F_{#1,\square,5}^{[40]}}
\newcommand{\Fsstar}[1]{F_{\square,k}^{[40],#1}}
\newcommand{\Ffivestars}[1]{F_{\square,5}^{[40],#1}}

\newcommand{\Wd}{W_{\Delta}^{[40]}}
\newcommand{\Wds}[1]{W_{\Delta}^{[40],#1}}
\newcommand{\Wdk}{W_{\Delta,k}^{[40]}}
\newcommand{\Wdkfive}{W_{\Delta,5}^{[40]}}
\newcommand{\Wdks}[1]{W_{\Delta,k}^{[40],#1}}
\newcommand{\Wdkfives}[1]{W_{\Delta,5}^{[40],#1}}

\newcommand{\Ws}{W_{\square}^{[40]}}
\newcommand{\Wss}[1]{W_{\square}^{[40],#1}}
\newcommand{\Wsk}{W_{\square,k}^{[40]}}
\newcommand{\Wskfive}{W_{\square,5}^{[40]}}
\newcommand{\Wsks}[1]{W_{\square,k}^{[40],#1}}
\newcommand{\Wskfives}[1]{W_{\square,5}^{[40],#1}}

\begin{document}
\maketitle

\begin{abstract}
\noindent  We consider the multi-objective mean-variance-skewness-kurtosis (MVSK) problem in portfolio selection, with and without shorting and leverage.
Additionally, we define a sparse variant of MVSK where feasible portfolios have supports contained in a chosen class of sets.
To find the MVSK problem's Pareto front, we linearly scalarize the four objectives of MVSK into a scalar-valued degree four polynomial $F_\lambda$ depending on some hyper-parameter $\lambda \in \Delta^4$.
As one of our main results, we identify a set of hyper-parameters for which $F_\lambda$ is convex over the probability simplex (or over the cube).
By exploiting the convexity and neatness of the scalarization, we can compute part of the Pareto front.
We compute an optimizer of the scalarization $F_\lambda$ for each $\lambda$ in a grid sampling of $\Delta^4$. 
To see each optimizer's quality, we plot scaled portfolio objective values against hyper-parameters. 
Doing so, we reveal a sub-set of optimizers that provide a superior trade-off among the four objectives in MVSK. 
\end{abstract}

\section{Introduction} \label{sec_intro}
We gently introduce the reader to the well-studied portfolio selection problem in finance. We progressively extend the model to include higher-order moments, shorting and leveraged positions, and sparsity.
 We explore some results of multi-objective optimization to apply the results later in \cref{sec_model}.

\subsection{The portfolio selection problem}
In finance, \emph{portfolio selection} is the task of selecting from a pool of available assets a weighted subset, called a \emph{portfolio}, in such a way that the portfolio maximizes return on investment while minimizing the risk of loss of capital. In 1952  Harry Markowitz \cite{marko93748} systematized portfolio selection by phrasing it as a constrained quadratic optimization problem involving data on past returns. Markowitz modeled the portfolio's profitability by the \emph{mean returns}, and he modeled risk by using the \emph{variance} as a proxy. Hence the model is a bi-objective optimization problem over the simplex:
\begin{equation} \label{Markowitz}
\begin{split} 
&\max ~ w^TM  \\
&\min ~  w^TVw \\
&\mathrm{s.t.}~w \in \Delta^n,
\end{split} 
\end{equation}
where $\Delta^n$ is the standard simplex, $V \in \R^{n\times n}$ is a covariance matrix, and $M \in \R^{n}$ is the vector of means.
Here $n$ denotes the number of assets that are available for selection and,  for each $i \in [n]$, $w_i$ denotes the weight of the i$^{\mathrm{th}}$ asset in the portfolio $w$.
Problem (\ref{Markowitz}) can be converted into a single-objective optimization problem of the following form
\begin{equation} \label{Markowitz_scal}
\begin{split} 
&\min ~(1-\lambda)w^TVw - \lambda w^TM \\
&\mathrm{s.t.}~w \in \Delta^n,
\end{split} 
\end{equation}
for some hyper-parameter $\lambda \in [0,1]$ modelling the investors risk tolerance.
 Hence, there are two conflicting objectives, maximizing the mean returns $w^TM$ and minimizing the variance $w^TVw$ in the returns. Though the model may seem crude by modern standards, it began the field of portfolio optimization, see \cite{bestPortfolioOptimization2010}. Variants of this model are still used and studied, see, for example, \cite{Lai2006MeanVarianceSkewnessKurtosisbasedPO, maringerGlobalOptimizationHigher2009, RePEc:com:wpaper:021}. We will consider an extension of (\ref{Markowitz}) before stating our contribution to the field.

\subsection{Extending to higher order moments}
The Markowitz model is often called the \emph{mean-variance framework}, as it only uses the mean and the variance in describing the problem. In statistics, the mean and variance are, respectively, called the first and second moments of the data. Using only the first two moments, the Markowitz model implicitly assumes that the data comes from a Gaussian distribution, however, this is not the case in practice, see \cite{SAMUELSON1975215, 10.2307/2331046}. Assuming the data is Gaussian, one underestimates the frequency of extreme events, like rare but significant losses. To account for this underestimation, several authors propose extending the model to include higher order moments like \emph{skewness} and \emph{kurtosis}, see \cite{krausSkewnessPreferenceValuation1976, Zhou2021SolvingHP, Mhiri2010InternationalPO, RePEc:com:wpaper:021, Shengzhi2011SemidefinitePR}.

Skewness, the third data moment, represents the asymmetric characteristics of a distribution. One can think of a Gaussian leaning in a particular direction, the skewness quantifying the direction and intensity of the leaning.
Positive skewness implies higher chances of occasional rare high returns, while
negative skewness implies the potential risk of occasional significant losses, see \cref{skew}.\\

\begin{figure}[ht]
\centering
	\includegraphics[scale=0.35]{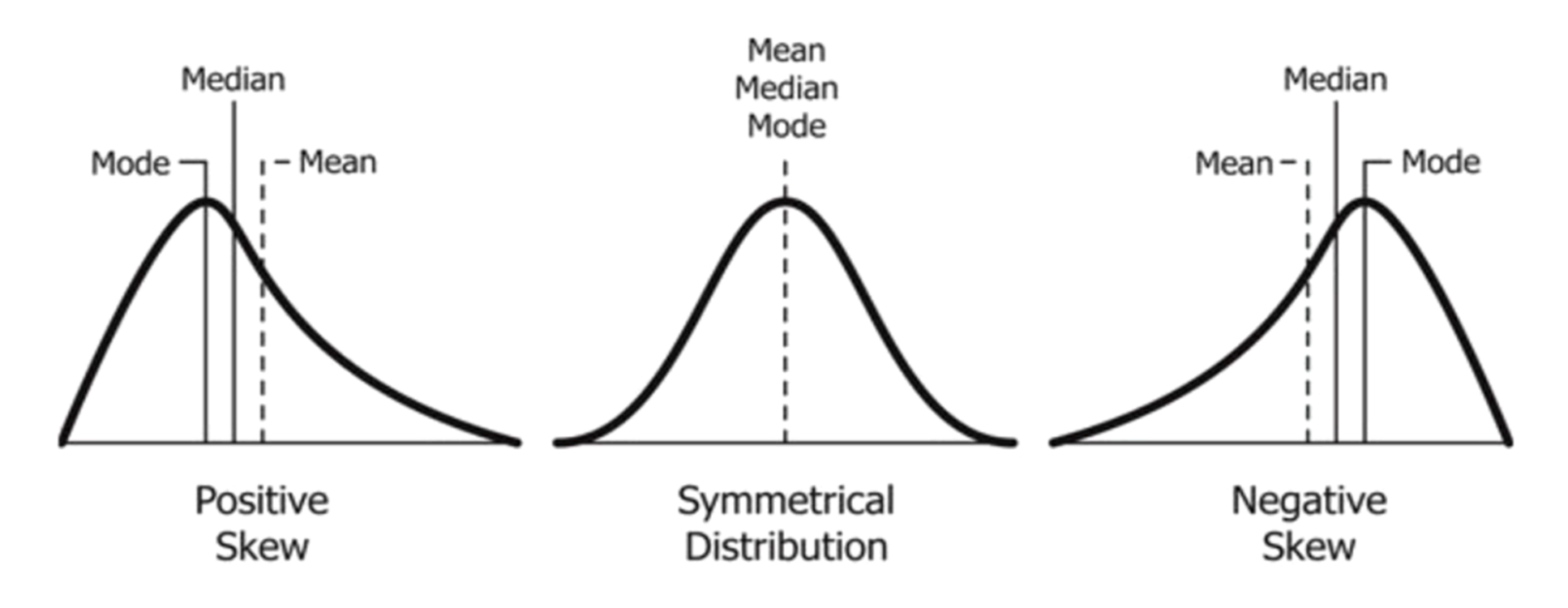}
	\caption{Diagram of the probability density function of a Gaussian with skewness, taken from \cite{doaneMeasuringSkewnessForgotten2011}.} 
	\label{skew}
\end{figure}

Kurtosis, the fourth data moment, is similar to variance in that larger values correspond to a sharper peak and fatter tails, i.e., more extreme returns on either side of the mean, compared with the normal distribution, see \cref{kurt}.\\

Alternatively, some researchers use notions of entropy and mutual information instead of correlation to quantify risk \cite{GONCALVES2022200101}. 

\begin{figure}[ht]
\centering
	\includegraphics[scale=0.8]{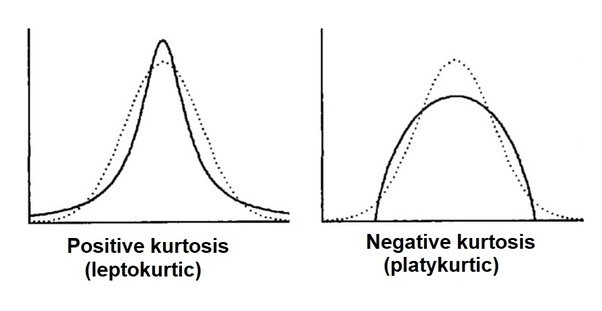}
	\caption{Diagram of distributions with positive and negative kurtosis compared to the normal distribution (dashed line), taken from \cite{DeCarlo1997OnTM}.} 
	\label{kurt}
\end{figure}

Most investors would prefer a large positive skewness and a small kurtosis if given a choice. Adding these new terms improves the model's expressiveness at the cost of adding more complexity. Skewness, in particular, is likely non-convex. We will elaborate more on the objective functions in \cref{subsec_MVSK}.\\

This extended model is called the \emph{mean-variance-skewness-kurtosis} (MVSK) problem. We give the mathematical description in \cref{subsec_MVSK}. It is a \emph{multi-objective optimization problem} (MOOP) with the first four moment functions as objectives, see (\ref{mean_fun}) and (\ref{mom_funs}) for the formal definitions. Solving the MVSK and some of its variants will be our primary task in this paper. We mention extensions to even higher-order moments in \cref{sec_conc}.

\subsection{Including leverage and shorting}\label{sub_lev_short}
The MVSK model, though an improvement on the mean-variance model, does not fully capture the richness of portfolio selection in the financial industry. We consider two further extensions to the model. Note that each extension is stand-alone and complementary to the rest. Hence they can be studied separately and used in combination. The first extension we consider is the option to hold \emph{shorted} and \emph{leveraged} positions, this means that the portfolio can consist of borrowed assets. Financial details aside, we allow  some of the assets in the portfolio to have negative weights or weights exceeding one. Mathematically, we no longer optimize over the simplex but rather over a bounded cube. Secondly, we consider \emph{sparsity}. A portfolio is sparse if it supports fewer assets than the selection pool allows. Given two equally well-performing portfolios, one often prefers the sparser portfolio to dense portfolios (where all weights are non-zero) because having fewer assets to manage leads to fewer transaction costs when rebalancing, see for example \cite{bestPortfolioOptimization2010}. Alternatively, some authors explicitly model transaction costs directly, as done in \cite{Li2019HighOP}.

%

\subsection{Multi-objective optimization problems} \label{moop_basic}
In addition to inheriting the difficulties of single-objective optimization problems, MOOPs have new challenges to address, see, e.g., \cite{ehrgottMulticriteriaOptimization2005} for background.
Consider the general MOOP:
\begin{equation} \label{def_moop}
\begin{split} 
\min f(x) &:= (f_1(x),f_2(x),...,f_p(x)) \\
\mathrm{ s.t.}~~&  x \in X ,
\end{split}
 \end{equation}
where $f_1,f_2,...,f_p$ are some scalar valued functions acting on $\R^n$, and $X \subseteq \R^n$.

How one defines optimality is the first change from single to multiple objectives. Real numbers are well ordered by $\leq$, and as such, it is clear when one solution gives a better objective value than another. In contradistinction, the values of MOOPs are real-valued vectors, and thus they are only partially ordered by $\leq$, applied entry-wise between two vectors (of equal size). Optimal solutions to MOOPs are hence only optimal in the sense of not being strictly worse than any other solution vector. Formally we define a partial order on vectors $v,w \in \R^p$:
\begin{equation} \label{def_partial_ord}
\begin{split} 
v \geq w \iff& v_i \geq w_i ~~ (i \in [p]), \\
v \gneqq w \iff& v \neq w~,~v_i \geq w_i ~~ (i \in [p]), \\
v  > w \iff& ~v_i > w_i ~~ (i \in [p]). \\
\end{split}
 \end{equation}
A point $x \in X$  is said to be \emph{Pareto optimal} for  (\ref{def_moop}) if there exists no $y \in X$ such that
\begin{equation}\label{pardef}
f(x) = (f_1(x),f_2(x),...,f_p(x)) \gneqq (f_1(y),f_2(y),...,f_p(y)) = f(y).
\end{equation}
Similarly, a  point $x \in X$  is said to be \emph{locally Pareto optimal} for (\ref{def_moop}) if it is Pareto optimal in some open neighborhood of $x$. The \emph{Pareto front} of (\ref{def_moop}) is defined as the set of all Pareto optimal solutions of (\ref{def_moop}). The following is a well-known fact.

\begin{lemma}
Consider the MOOP (\ref{def_moop}), and assume that the objectives $f_1,f_2,...,f_p$ are all convex functions and that the domain $X$ is a convex set. Then
any locally Pareto optimal point $x$ of (\ref{def_moop}) is also (globally) Pareto optimal.
\end{lemma}
\begin{proof}
Suppose by way of contradiction that  $x$ is not globally Pareto and let $y \in X$  be a point such that $f(y) \lneqq  f(x)$.
Take any $t \in (0,1)$ and observe that via convexity we have
$$
f(ty + (1-t)x) \leq  tf(y) + (1-t)f(x) \lneqq   tf(x) + (1-t)f(x) =  f(x).
$$
Since this holds for arbitrarily small positive values of $t$ it holds that $x^*$ is not locally Pareto optimal, contradicting our initial assumption.
\end{proof}

Among the several approaches to optimizing a MOOP, we will be looking for optimizers via \emph{scalarizations} of the MOOP. Scalarization is a well-known approach that converts a MOOP into a single objective optimization problem called the scalarized problem. Several authors have done this for MVSK by encoding some objectives as constraints, see, e.g., \cite{maringerGlobalOptimizationHigher2009}. One downside of this approach is that one must make an a priori estimate of these objectives. Alternatively, one can scalarize by combining the multiple objectives into a single scalar-valued objective function. We follow this approach. For the MVSK problem, the literature predominantly considers two scalarizations.
The first is the Minkowski scalarization, as seen in \cite{Lai2006MeanVarianceSkewnessKurtosisbasedPO, Mhiri2010InternationalPO, Aracolu2011MeanVarianceSkewnessKurtosisAT}. Here one first computes the optimal values for each of the objectives independent of the others
$$
f_i^* :=  \min_{x \in X} f_i(x) ~~ (i \in p).  
$$
Using these independent optima one constructs, for some positive user-defined hyper-parameter $0 < \lambda \in \R^p$, the Minkowski distance scalarization is as follows:
\begin{equation} \label{minkow_scal}
\begin{split} 
\min_{x\in X} \sum_{i \in [p]}\big| f_i(x) -f_i^* \big|^{\lambda_i}.
\\
\end{split}
\end{equation}

The second scalarization is simply a linear combination of the objectives with the linear weights being some choice of hyper-parameter $\lambda \in \R^p$, see \cite{RePEc:com:wpaper:021, Kleniati2009PartitioningPF}.
 The resulting scalarized optimization problem is hence
\begin{equation} \label{eq_lin_scal}
\min_{x\in X} F_\lambda(x),
\end{equation}
where 
\begin{equation} \label{F_lambda}
F_\lambda(x) := \sum_{i \in [p]}\lambda_i f_i(x).
\end{equation}

Note that this scalarization has a linear dependence on hyper-parameters and is also conceptually simple to interpret. We will be using this linear scalarization throughout this paper. Optimizers of the scalarized problem are not guaranteed to be Pareto optimal for the MOOP, but for \emph{neat} scalarizations, this is the case.  A scalarization is said to be \emph{neat} if any optimal solution $w$ of the scalarized problem is also a Pareto optimal solution of the original MOOP.

\begin{lemma}[Proposition 3.9  \cite{ehrgottMulticriteriaOptimization2005}] \label{lem_neat}
If $\lambda > 0$, then the scalarization (\ref{eq_lin_scal}) is neat, i.e., global optimizers of (\ref{eq_lin_scal}) are (global) Pareto optimizers of (\ref{def_moop}).
\end{lemma}
\begin{proof}
Let $x\in X$ be an optimal solution of (\ref{eq_lin_scal}) and suppose by way of contradiction that $x$ is not Pareto optimal for (\ref{def_moop}), i.e., there exists a $y \in X$ such that $ f(x) \gneqq f(y)$. Then 
$$
 \sum_{i \in [p]}\lambda_i f_i(x) \gneqq  \sum_{i \in [p]}\lambda_i f_i(y)
$$
because $\lambda_i > 0$ for all $i \in [p]$. Hence, this contradicts the fact that $x$ optimizes (\ref{eq_lin_scal}).
\end{proof}

Consider now the case when the feasible set is defined as 
\begin{equation} \label{def_X}
X := \{x \in \R^n : g_j(x) \geq 0,~ j \in [q]\},
\end{equation}
for some functions $g_1,g_2,...,g_q: \R^n \to \R$.
Let $J(x)= \{ j \in [q] : g_j(x) = 0 \}$ denote the index set of active constraints. The following result holds for the MOOP (\ref{def_moop}).
\begin{theorem} [Theorem 3.25 \cite{ehrgottMulticriteriaOptimization2005}] \label{theo_KKT}
Let $X$ be a set as defined in (\ref{def_X}).
Let $f_1,f_2,...,f_p$, $g_1,g_2,...,g_q$ be scalar-valued functions that are continuously differentiable at $x^* \in X$.
Assume that $x^*$ is a Pareto optimal point of (\ref{def_moop}) and that there is no vector $v \in \R^n$ such that
\begin{equation} \label{KT_cond}
\begin{split} 
\inp{\nabla f_i(x^*), v}& \leq 0 ~\text{for all}~ i \in [p], \\
\inp{\nabla f_k(x^*), v}& < 0 ~\text{for some}~ k \in [p], \\
\inp{\nabla f_j(x^*), v}& \leq 0 ~\text{for all}~ j \in J(x^*). \\
\end{split}  
\end{equation}
Then there exist vectors $\lambda \in \Delta^p$ and $\eta \in \bb{R}^q$  such that $\lambda > 0$, $\eta \geq 0$, and
\begin{gather*} 
\sum_{i \in [p]} \lambda_i \nabla f_i(x^*) + \sum_{j \in [q]} \eta_j \nabla g_j(x^*) = 0, \\
\sum_{j \in [q]} \eta_j g_j(x^*) = 0.
 \end{gather*}
Therefore, $x^*$ is a KKT point of the following scalarization of  problem (\ref{def_moop}):
\begin{equation} \label{scal}
\begin{split} 
&\min ~~   F_\lambda(x) := \sum_{i \in [p]}\lambda_i f_i(x) \\
&\text{ s.t. } ~~ x \in X.
\end{split}
\end{equation}
\end{theorem}

A point $x^* \in X$ that is Pareto optimal and satisfies the system (\ref{KT_cond}) is also known in the literature as being \emph{properly efficient in the Kuhn-Tucker sense} (see Definition 2.49 in \cite{ehrgottMulticriteriaOptimization2005}).


\begin{proposition} \label{prop_KKT_convex}
Assume the conditions of \cref{theo_KKT} hold. If in addition $X$ is a convex set and  $F_\lambda$ a convex function, then $x^*$ is a global optimizer of (\ref{scal}).
\end{proposition}
\begin{proof}
The claim follows from the fact that any KKT point of a convex problem must be a global optimizer. 
\end{proof}

Let us again consider the scalarized problem (\ref{eq_lin_scal}) where $F_\lambda(x) := \sum_{i \in [p]}\lambda_i f_i(x)$ for some $0 < \lambda \in \R^p$. Depending on the functions $f_1,...,f_p$, the hyper-parameter $\lambda$, and the domain $X$, problem (\ref{eq_lin_scal}) can still be extremely difficult to solve. However, in the special case when the objective $F_\lambda(x)$ and the domain $X$ are convex (strictly convex), there are efficient methods to find the (unique) minimizer \cite{doi:10.1137/1.9781611974997}. Having found an optimizer to the scalarized problem, \cref{lem_neat} relates said optimizer back to a Pareto point of the MOOP.\\

The core theme of this paper is to partially recover the Pareto set of the MVSK problem by solving different linear scalarizations of the MVSK problem. In order to achieve this we identify classes of hyper-parameters $\lambda \in \Delta^4$ that ensure the resulting scalarization $F_\lambda$ is convex over the domain of optimization, which is either the standard simplex or the cube.  \\

\subsection{Structure of the paper}

In \cref{sec_model} we formally derive the MVSK  problem using a random variable to model asset price data, see (\ref{MVSK}).
Applying the linear scalarization, we get a problem of the form (\ref{eq_lin_scal}). If we take $\lambda \in \Delta^4,~\lambda_4 \neq 0$, then $F_\lambda$ becomes a quartic polynomial.
For our purposes, the domain $X$ can be either the simplex or the cube.
Polynomial optimization is an active field of research, for a general reference on this topic we refer to \cite{doi:10.1142/p665}.
For general polynomial optimization, quadratic polynomials are already hard to optimize over the simplex, see \cite{Motzkin1965MaximaFG}.
However, in our particular setting, we will show that for some $\lambda$'s, the problem is convex and can be efficiently solved.
The section ends with a suggested sparse variant of MVSK.\\

In \cref{sec_scalar_MVSK} we characterize the parameters $\lambda \in \Delta^4$ for which $F_\lambda$ is convex. Thus we obtain, for some scalarizations, a differentiable convex problem of the form (\ref{eq_lin_scal}). We can apply first-order optimization methods to this problem and obtain good approximations of optimizers. The particular algorithm we use is called FISTA, described in more detail in \cref{sub_sec_fista}. For the scalarizations $F_\lambda$ that are not convex, we still apply FISTA, though we only have that the returned optimizer is locally optimal. Regardless, a local optimizer of (\ref{eq_lin_scal}) is still a locally Pareto optimal solution to (\ref{eq_lin_scal}), provided $\lambda_i > 0$ for all $i \in [4]$.\\

We partially recover the Pareto set of MVSK by solving different scalarized MVSK problems.
In \cref{sec_num}, we demonstrate the efficacy of our approach by conducting a numerical experiment on real-world data. We explain the well-known optimization algorithm FISTA that we use, and we list its beneficial properties.
Four optimization domains are considered: the simplex, the cube, and their respective sparse analogs. 
We visualize and compare the resulting approximate Pareto sets.
A striking feature of our numerical results is that although each Pareto point obtained in this way is not strictly worse than any other in the sense of the partial order described in (\ref{def_partial_ord}), there are some points that provide a better trade-off among the four objectives. We call such points \textit{solutions of superior trade-off}. These solutions of superior trade-off are described in \cref{sub_sec_grid} and visualized in \cref{sub_sec_num_res}.

 Finally we conclude with some remarks on even higher order models and alternative scalarizations in \cref{sec_conc}.
\\

Hence, our contributions in this paper are three-fold. We characterize some scalarizations of the MVSK problem that result in a convex problem; we propose a sparse MVSK problem; we show that some Pareto optimizers give a better trade-off among the different objectives.

\section{The MVSK model} \label{sec_model}
This section gives the mathematical formulation of the MVSK optimization problem. We start with a random variable representing asset prices and then define the four objectives that constitute the MVSK MOOP. Finally, we conclude with a proposed sparse variant of MVSK and a list of some motivations why this model is of practical interest.

\subsection{The domain of optimization}
A \emph{portfolio} consists of a weighted selection of $n \in \bb{N}$ assets, represented by $w \in \R^{n}$. We will consider two choices of the domain for portfolios. 

The first is the \emph{standard simplex}.
In this setting, investors cannot ``short sell" assets nor take ``leveraged" positions.
We write $w \in \Delta^n$, where $\Delta^n := \{w \in [0,1]^n : \sum_{i=1}^n w_i = 1\}$.

\textit{Short selling} is the act of selling a borrowed asset with the hope that it will depreciate over time. 
After the asset has lowered in price, one repurchases the asset to return to the original owner. 
Mathematically, negative portfolio weights can model this, i.e., $w \in (-\infty,1]^n$.
In finance, \textit{leverage} is buying additional assets on credit, and the investor then aims to make an additional profit using the extra assets before repaying the debt.
Mathematically this means that the portfolio weights can sum up to more than one.

The second setting is the \emph{cube}, where we allow short selling and leverage.
There is a bound $B\in \R_+$ on how leveraged a position can be.
Mathematically we write $w \in [-B, B]^n $, where we set $B=1$ for simplicity.
For a general overview of financial terms, we refer the reader to any standard text like \cite{Melicher1996FinanceIT}.\\

To distinguish the general results of \cref{sec_intro} from the particular setting of MVSK, we change the notation
 from $x$ to $w$. With this, we intend the reader to not think of a general vector $x \in \R^n$ anymore but rather a vector of weights, $w \in \Delta^n$ or $w \in [-1, 1]^n$.

\subsection{Computing the objective functions}
In portfolio optimization, the underlying assets are usually paper assets like stock in publicly traded companies. We do not restrict ourselves to this setting, but it is a useful example of the model. For asset $i\in [n]$ let $\wtl R_i$ denote the  \emph{relative return}, a random variable taking values in $\bb{R}^n$. Hence, $\wtl R_i$ is the fractional change in value relative to the initial cost of purchasing the asset $i$. Let 
\begin{equation}\label{mean} 
M := \bb{E}[\wtl R] = (\bb{E}[ \wtl R_i])_{i \in [n]} \in \R^n,
\end{equation} 
denote the vector of expected returns.
Now define the following random variable:
$$
R  := \Big{(} {\wtl R_i - M_i}\Big{)}_{i \in [n]},
$$
which  we call the \emph{centralized relative returns}.

We can now begin to build the multi-objective MVSK problem. To start, we define the statistical moments as functions of the weights $w$ and the data $R$. Let
\begin{equation}\label{mean_fun}  
f_1(w):=   M^Tw,
\end{equation}
represent the expected returns of the portfolio. We do not use $R$ as it has a zero mean. Similarly, we can define for $k = 2,3,4$ the functions
\begin{equation}\label{mom_funs_k} 
f_k(w) := \bb{E}[ \inp{R,w}^{k} ] ~,~ \inp{R,w}^{k} := (R^Tw)^k = \Big(\sum_{i\in [n]}R_iw_i  \Big)^k .
\end{equation}
These functions are related to the second, third, and fourth moments of $R$ as follows:
\begin{equation}\label{mom_funs} 
\begin{split} 
&f_2(w) :=  w^TV w,  \\
&f_3(w) := (w \otimes w)^TS w, \\
&f_4(w) := (w \otimes w)^TK (w \otimes w), \\
\end{split}
\end{equation} 
where $V := \bb{E}[RR^T] \in \R^{n\times n}$ is the \emph{covariance matrix},   $S := \bb{E}[(R \otimes R) R^T]  \in \R^{n^2 \times n}$ is the \emph{skewness matrix}, and $K :=  \bb{E}[(R \otimes R)(R \otimes R)^T] \in \R^{n^2 \times n^2}$ is the \emph{kurtosis matrix}, all w.r.t. the data $R$. With slight abuse of terminology, we refer to $f_2(w)$ as the variance of portfolio $w$, and similarly, $f_3(w)$ and $f_4(w)$ are called its skewness and kurtosis.

Let us examine these functions. Firstly we note that $f_1,f_2$, and $f_4$ are convex. Indeed $f_1$ is linear and therefore convex. The functions $f_2$ and $f_4$ are convex and nonnegative for all $w \in \bb{R}^n$ by virtue of $V$ and $K$ being positive semidefinite. To see why $V$ is PSD observe that $V$ is the expectation of a random variable $RR^T$ taking PSD matrices as values. Hence, the Hessian of $f_2$, $H(f_2)= V$, is PSD, and thus $f_2$ is convex. The convexity and nonnegativity of the kurtosis function $f_4$ follows for a similar reason.

\subsection{MVSK optimization problem} \label{subsec_MVSK}
With the individual objective functions defined in (\ref{mean_fun}) and (\ref{mom_funs}), we can define the MVSK MOOP:
\begin{equation}\label{MVSK} 
\begin{split} 
\max & ~~ f_1(w)  \\
\min & ~~ f_2(w)  \\
\max & ~~ f_3(w)  \\
\min & ~~ f_4(w) \\
s.t. &~~  w \in \Delta^n.\\
\end{split}
\end{equation}
Interpret this program as follows: one wishes to maximize returns while minimizing extreme events like rare but significant losses. The ``odd" functions $f_1$ and $f_3$ correspond, in expectation, to increased returns when positive and to losses when negative. While the ``even" functions $f_2$ and $f_4$ describe the spread of returns, with larger values corresponding to more significant fluctuations at the extremes. Note that variance and kurtosis are symmetric, which means they treat extreme profits and losses with equal prejudice. As a rule of thumb, investors prefer consistently high returns and dislike volatility.

In contradistinction to scalar optimization problems with a single scalar optimal value, a MOOP has a set of Pareto optimal points, sometimes called its \emph{efficient frontiers} or \emph{Pareto front}. From the investor's perspective, one need only choose from this frontier to be sure that no strictly better choice exists. Some Pareto solutions provide a better spread among the multiple objectives than others, more on this in \cref{sec_num}. However, it still falls to the investor to decide how to choose among these solutions. Hence our task is to find this efficient frontier of (\ref{MVSK}), but first, let us consider some extensions.

\subsubsection{A sparse variant of MVSK}
In selecting a portfolio, we prefer sparse weights, that is, portfolios with most weights equal to zero. Sparse should be understood in contradistinction to dense (portfolios), where almost all of the $n$ possible asset choices have a non-zero weight assigned. Having more assets beyond a point of ``reasonable diversification" could increase management fees and transaction costs, as rebalancing the portfolio requires adjusting more weights. The additional costs will then counteract the profitability of the portfolio.
A second reason a portfolio can become sparse is by disallowing certain asset combinations. When one knows that two assets are causally linked, the portfolio gains little \emph{diversification} by holding both. One of the core ideas in portfolio selection, diversification, is the principle that buying causally unrelated stocks will protect the portfolio from the possibility of significant losses. The idea is that one expects the depreciation of a single stock to be unrelated (or inversely related) to the value of other stocks. Of course, this only holds outside of systemic events like economic crises, see \cite{taleb2020statistical}. Diversification is the motivation for why one does not simply invest all one's capital in the single asset showing the largest return. We reformulate a general sparse version of the problem (\ref{MVSK}) as follows:

\begin{equation}  \label{SMVSK}
\begin{split} 
&\max  ~~ f_1(w)  \\
&\min  ~~ f_2(w)  \\
&\max  ~~ f_3(w)  \\
&\min  ~~ f_4(w) \\
&s.t. ~~  w \in \Delta^n \\
&~~~~~ \prod_{i\in C}w_i = 0 ~\text{for}~C\in \cal{C}, 
\end{split}
\end{equation}
where $\cal{C}$ is some set of subsets of $[n]$.
There are two motivating instances for the above form of sparsity.\\

\textbf{Reducing transaction} costs and management fees by bounding the number of stocks in the portfolio. We do this by setting
$$
{ \cal{C}} = \{ C \subseteq [n] : |C|=k \},
$$
 for some integer $k \leq n$. This is equivalent to saying that the solution $w$ must not have more than $k-1$ non-zero entries, i.e., $|\supp(w)| \leq k-1$, where $\supp(w):= \{i \in [n]: w_i \neq 0\}$. In terms of the portfolio, this is equivalent to holding at most $k-1$ assets at any given time.\\
 
\textbf{Accounting for causally linked stocks}. To factor in the notion of diversification into the above model, we set 
$$
{\cal{C}} = \{(i,j) : i \neq j,~ |\corr(R_i, R_j)| \geq \gamma \},
$$
for some $\gamma$, where $\corr$ is the Pearson correlation coefficient, see \cite{wackerly2014mathematical}, but other notions of mutual information or expert opinion could also be used in constructing $\cal{C}$.\\

By adding sparsity, we are restricting the domain over which we optimize, and hence we obtain a possibly weaker optimal solution. Indeed, if $v$ is a Pareto optimal solution of the sparse problem (\ref{SMVSK}), it need not be Pareto for the dense problem (\ref{MVSK}).

\section{Scalarizing MVSK} \label{sec_scalar_MVSK}
In this section, we will consider a linear scalarization (\ref{LMVSK}) of the MOOP (\ref{MVSK}) and analyse the conditions under which the resulting objective $F_\lambda$ is convex. In particular we characterize the coefficients $\lambda$ for which $F_\lambda$ is convex over the simplex (resp., cube) in terms of the data $R$. In general, it is highly desirable for objective functions to be (strict) convex as it ensures that any local optimizer is also a (unique) global optimizer. Convex functions are well studied and efficiently optimized if the gradient is known, see, for example, the standard textbook \cite{boyd_vandenberghe_2004}.

For any choice of $\lambda := (\lambda_1, \lambda_2, \lambda_3, \lambda_4)  \geq 0$, consider the following scalarization of (\ref{MVSK}):
\begin{equation}\label{LMVSK} 
\begin{split} 
F_\lambda^*:=\min & ~~ F_\lambda(w) := -\lambda_1 f_1(w) + \lambda_2 f_2(w) -\lambda_3 f_3(w) + \lambda_4 f_4(w) \\
s.t. &~~  w \in \Delta^n. \\
 \end{split}
\end{equation}
 This linear formulation is often called the ``weighted sum method". Recall that we are looking for the optimizers of $F_\lambda$  and not for the optimal value. Since for any scalar $c > 0$ we have
$$\argmin_{w \in \Delta^n}F_\lambda(w) =  \argmin_{w \in \Delta^n}F_{c\lambda}(w),$$
we can without loss of generality scale $\lambda$ to lie in the simplex $ \Delta^4$.

Our ambition is hence as follows:
Via \cref{lem_neat}, we can find Pareto optimal solutions of (\ref{MVSK}) by solving (\ref{LMVSK}) for $\lambda > 0$ such that $F_\lambda$ is convex. By doing this for various appropriate $\lambda$, we hope to recover part of the Pareto front. Later we also apply the same process for $\lambda$ that are neither strictly positive nor resulting in convex $F_\lambda$, this still yields a local optimizer of (\ref{LMVSK}) but we have no guarantees of it being Pareto optimal for the MOOP (\ref{MVSK}).

It is not necessarily true that all Pareto optimizers of (\ref{MVSK}) are also optimizers for some scalarization of the form (\ref{LMVSK}).
It is true, however, that each Pareto point $x^*$ of (\ref{MVSK}) satisfying (\ref{KT_cond}) corresponds to a Karush–Kuhn–Tucker (KKT) point of some scalarization with $\lambda > 0$, as was seen in \cref{theo_KKT}.
Applying \cref{theo_KKT} to our setting, we have the following result.

\begin{corollary}[KKT] \label{every}
For every Pareto optimal point $w^*$ of (\ref{MVSK}) satisfying the conditions (\ref{KT_cond}), there exists a  positive $\lambda  \in \Delta^4$, $\lambda >0$, for which $w^*$ is also a KKT point of (\ref{LMVSK}).
\end{corollary}

We now shift to finding $\lambda$ for which $F_\lambda$ is convex over the simplex, the cube, or the whole space $\R^n$.

\subsection{Convex scalarization of MVSK} \label{convex_res}
In general, optimizing a quadratic polynomial over the simplex is already hard. Indeed, recall the Motzkin-Straus \cite{Motzkin1965MaximaFG} formulation of the stability number of an undirected graph. Problem (\ref{LMVSK}) has a quartic objective and is expected to contain the difficulty of the quadratic case. However, we can solve some convex problems efficiently using gradient methods \cite{boyd_vandenberghe_2004}. We now give several characterizations of $\lambda \in \Delta^4$ for which $F_\lambda(w)$ is convex. Begin by considering the gradient of $ F_\lambda$ at a point $w$
\begin{equation*} 
\begin{split} 
\nabla  F_\lambda(w)  &= \nabla (-\lambda_1M^Tw + \lambda_2\bb{E}[ \inp{R,w}^2 ] - \lambda_3\bb{E}[ \inp{R,w}^{3} ]  + \lambda_4\bb{E}[ \inp{R,w}^{4} ]  ) \\
&=-\lambda_1M + 2\lambda_2\bb{E}[R\inp{R,w} ] - 3\lambda_3\bb{E}[R \inp{R,w}^{2} ]  + 4\lambda_4\bb{E}[R \inp{R,w}^{3} ].  \\
 \end{split}
 \end{equation*}
The Hessian of $F_\lambda$ at $w$ is given by
\begin{equation} 
\begin{split} \label{Hess} 
H(F_\lambda)(w) &= \nabla^2  F_\lambda(w) = 2\lambda_2\bb{E}[RR^T] - 6\lambda_3\bb{E}[RR^T \inp{R,w} ]  + 12\lambda_4\bb{E}[RR^T \inp{R,w}^{2} ]  \\
&= \bb{E}[ ( 2\lambda_2 - 6\lambda_3\inp{R,w} + 12\lambda_4\inp{R,w}^{2})RR^T] \\
&= \bb{E}[2\qua RR^T], 
 \end{split}
 \end{equation}
where we define
\begin{equation} \label{Phi}
\qua := 6\lambda_4\inp{R,w}^{2} - 3\lambda_3\inp{R,w} + \lambda_2.
\end{equation}
Define the quadratic polynomial
\begin{equation} \label{Psi}
\quay := 6\lambda_4y^{2} - 3\lambda_3y + \lambda_2,
\end{equation}
and observe that $\Psi_\lambda (\inp{R,w}) = \qua$ under the change of variables $y:= \inp{R,w}$.

\begin{lemma} \label{pos_coef}
Let $\Rtriu := \min_{w \in \Delta^n}\inp{R,w}$, $\Rtri := \max_{w \in \Delta^n}\inp{R,w}$. Then $\qua \geq 0$ for all $w \in \Delta^n$ if and only if one of the following conditions is satisfied:
\begin{itemize}
\item[(i)]   $\lambda_4 = 0$ and $3\Rtri \lambda_3 \leq \lambda_2$,
\item[(ii)]  $\lambda_4 > 0$ and $\lambda_3 \leq \sqrt{\frac{8}{3}\lambda_2\lambda_4}$,
\item[(iii)] $\lambda_4 > 0$, $\lambda_3 > \sqrt{\frac{8}{3}\lambda_2\lambda_4}$, $ 3\Rtri \lambda_3 \leq \lambda_2+6\Rtri ^2 \lambda_4$, and  $4 \Rtri \lambda_4  \leq \lambda_3$,
\item[(iv)]  $\lambda_4 > 0$,  $\lambda_3 > \sqrt{\frac{8}{3}\lambda_2\lambda_4}$, 
$3\Rtriu \lambda_3 \leq  \lambda_2+6 \Rtriu^2 \lambda_4$, and  $4  \Rtriu \lambda_4 \geq  \lambda_3$.

\end{itemize}
\end{lemma} 
\begin{proof}
If $\lambda_4 = 0$, then  $\qua \geq 0$ if and only if $3\inp{R,w}\lambda_3 \leq \lambda_2 $.
Requiring that this hold for all $w \in \Delta^n$ is equivalent to requiring $3 \Rtri\lambda_3 \leq \lambda_2$. So we find case (i).

Suppose $\lambda_4 > 0$ and consider the discriminant $\Delta_\lambda :=  9\lambda_3^2 -24 \lambda_2\lambda_4$ of $\quay$.

If $\Delta_\lambda < 0$ then $\quay$ has no real roots, meaning that $\quay > 0$ for all $y \in \R$. The condition $\Delta_\lambda < 0$ is equivalent to requiring $\lambda_3 < \sqrt{\frac{8}{3}\lambda_2\lambda_4}$. In the case that $\Delta_\lambda = 0$ then $\quay$ has double root at $y = \frac{3 \lambda_3 }{12 \lambda_4}$ and $\quay \geq 0$ for all $y \in \R$. So we find case (ii).

Assume $\Delta_\lambda > 0$. Then $\quay$ has two roots 
$$
y_l := \frac{3 \lambda_3 - \sqrt{\Delta_\lambda} }{12 \lambda_4},~  y_u := \frac{3 \lambda_3 + \sqrt{\Delta_\lambda} }{12 \lambda_4}.
$$
Hence there are only two cases when $\qua \geq 0$ for all $w \in \Delta^n$. 
The first is when all values of $y = \inp{R,w}$ are below $y_l$, i.e.,
\begin{gather*} 
\Rtri \leq  \frac{3 \lambda_3 - \sqrt{\Delta_\lambda}}{12 \lambda_4} = y_l\\
\iff \sqrt{\Delta_\lambda} \leq   3 \lambda_3 - 12 \Rtri \lambda_4 \\
\iff   \Delta_\lambda = 9\lambda_3^2 -24 \lambda_2\lambda_4 \leq  (3 \lambda_3 - 12 \Rtri \lambda_4)^2  ~\text{and}~    0 \leq   3 \lambda_3 - 12 \Rtri \lambda_4 \\
\iff   4 \lambda_4 \Rtri  \leq \lambda_3 ~\text{and}~  3\Rtri \lambda_3 \leq \lambda_2+6\lambda_4 \Rtri ^2.
\end{gather*} 
Hence, we have shown case (iii).
The second case is when  all values of $y= \inp{R,w}$ are  above $y_u$, i.e.,
\begin{gather*}
\Rtriu \geq \frac{3 \lambda_3 + \sqrt{\Delta_\lambda}}{12 \lambda_4} = y_u\\
\iff  12 \Rtriu \lambda_4 - 3 \lambda_3 \geq \sqrt{\Delta_\lambda}  \\
\iff    9\lambda_3^2 -24 \lambda_2\lambda_4 \leq  (3 \lambda_3 - 12 \Rtri \lambda_4)^2  ~\text{and}~ 0 \leq   12 \Rtriu \lambda_4 - 3 \lambda_3  \\
\iff   4 \Rtriu \lambda_4  \geq  \lambda_3 ~\text{and}~ \lambda_2+6\lambda_4 \Rtriu ^2 \geq 3\Rtriu \lambda_3.
\end{gather*} 
With this case (iv) is proved and the proof is concluded.
\end{proof}

\begin{observation}
Note that condition (iv) in \cref{pos_coef} implies $\Rtriu > 0$. In numerical experiments with real-world data, we often have 
$\Rtriu < 0$, and thus condition (iv) seldom holds.
\end{observation}

\begin{corollary} \label{cor_convex}
If $\lambda \in \Delta^4$ satisfies any of the conditions (i)-(iv) of \cref{pos_coef} then $F_\lambda$ is convex on $\Delta^n$.
Moreover if $\lambda \in \Delta^4$ satisfies the condition (ii) of \cref{pos_coef} then $F_\lambda$ is convex on $\R^n$.
\end{corollary}
\begin{proof}
These results follow directly from the fact that the Hessian is PSD when $\qua \geq 0$, i.e.,
$$
\qua \geq 0 \implies  H(F_\lambda)(w) \succeq 0.
$$
\end{proof}

The results of \cref{pos_coef} can be extended to strict convexity by making a timid assumption on the random variable $R$.

\begin{corollary}
Consider the Hessian given in (\ref{Hess}) for some $\lambda \in \Delta^4$. Assume that  $\bb{E}[RR^T] \succ 0$ and that, for all $w \in \Delta^n$,  $\qua > 0$ a.e.\footnote{The abbreviation \textit{a.e.} stands for \textit{almost everywhere} and is used to indicate that the accompanying statement may fail, but only on a set of measure zero.}. Then $F_\lambda$ is strictly convex on $\Delta^n$.
\end{corollary}

\begin{proof}
Since $\qua > 0$ a.e. for all $w \in \Delta^n$, we have that $H(F_\lambda) \succeq 0$ on $\Delta^n$.
Assume by way of contradiction that $H(F_\lambda)$ is not positive definite, then there exists a nonzero $v  \in \R^n \setminus \{0\}$ such that $v^T H(F_\lambda) v = v^T \bb{E}[\qua RR^T] v = 0.$ By linearity of the expectation this implies that
$\bb{E}[\qua v^TRR^Tv] = 0.$
Since each argument is a.e. nonnegative we have that $\qua v^TRR^Tv = 0 \text{~a.e.},$ and thus $v^TRR^Tv = 0 \text{~a.e.}$ by virtue of $\qua > 0$ a.e..
Taking the expectation we get $0 = \bb{E}[v^TRR^Tv] =  v^T\bb{E}[RR^T]v$ contradicting our assumption that $\bb{E}[RR^T] \succ 0$.
\end{proof} 
 
\begin{observation} \label{covex_lambda}
The scalarization $F_\lambda(w)$ is convex on the standard simplex $\Delta^n$ if and only if the hyper-parameter $\lambda = (\lambda_1,\lambda_2,\lambda_3,\lambda_4) \in \Delta^4$ satisfies
\begin{equation} \label{SDP_convex}
\lambda_3 \leq \max_{\gamma \geq 0}\{ \gamma  :\bb{E}[ ( 2\lambda_2 - 6\gamma \inp{R,w} + 12\lambda_4\inp{R,w}^{2})RR^T] \succeq 0 ~~(w \in \Delta^n)  \}.\\
\end{equation} 
\end{observation}
When $\lambda_3 = 0$, $F_\lambda$ is convex, so $\lambda_3$ is the limiting factor to PSDness of the Hessian of $F_\lambda$. Hence, we seek the largest $\lambda_3$ for which $H(F_\lambda)(w) \succeq 0$ for all $w \in \Delta^n$. The parameter $\lambda_1$ plays no role in the convexity of $F_\lambda$. The Hessian is linear in $\lambda$ but quadratic in $w$. The expression in problem (\ref{SDP_convex}) is not simply a \emph{linear matrix inequality}, and to the best of our knowledge, it cannot be solved efficiently \cite{doi:10.1137/1.9780898718829}. 


Thus far, we have considered convexity over the simplex domain. Analogous results hold for the cube. To generalize \cref{pos_coef} to the cube simply modify the bounds $\Rtri$ and $\Rtriu$ by defining
$$
\Rsqr := \max_{w \in [-1,1]^n}\inp{R,w},~ \Rsqru := \min_{w \in [-1,1]^n}\inp{R,w}.
$$


\subsubsection{Regions of hyper-parameters $\lambda$ for which $F_\lambda$ is convex}

We define the following nested sets of hyper-parameters $\lambda$
$$
\Lambda_{+} \subseteq  \Lambda_{\Delta} \subseteq  \wh{\Delta}~\text{and}~ \Lambda_{+} \subseteq \Lambda_{\square} \subseteq  \wh{\Delta},
$$
where 
\begin{equation} \label{conv_doms_3} 
\begin{split} 
&\wh{\Delta} := \{(\lambda_2, \lambda_3, \lambda_4) \geq 0 :  \lambda_2+ \lambda_3 + \lambda_4 \leq 1 \}  \subseteq  \R^3, \\
&\Lambda_{+} := \{(\lambda_2, \lambda_3, \lambda_4) \in \wh{\Delta}   : \lambda_2\lambda_4 \geq (3/8)\lambda_3^2 \}, \\
&\Lambda_{\Delta} :=\{(\lambda_2, \lambda_3, \lambda_4) \in \wh{\Delta} : \lambda \text{ satisfies any condition of \cref{pos_coef} for } \Rtri \text{ and } \Rtriu    \},\\
&\Lambda_{\square} :=\{ (\lambda_2, \lambda_3, \lambda_4) \in \wh{\Delta} : \lambda \text{ satisfies any condition of \cref{pos_coef} for }  \Rsqr \text{ and } \Rsqru   \}.\\
\end{split} 
\end{equation} 
Via \cref{pos_coef}, it now follows that if $\lambda \in \Lambda_{+}$ then $F_\lambda$ is convex over $\R^n$. Similarly if $\lambda \in \Lambda_{\Delta}$ (resp., $\Lambda_{\square}$) then $F_\lambda$ is convex over the simplex $\Delta^n$ (resp., the cube $ [-1,1]^n$).
The benefit of eliminating a variable ($\lambda_1$ in this case) is that the hyper-parameter sets $\Lambda_{+},~\Lambda_{\Delta},~\Lambda_{\square}$, and $\wh{\Delta}$ can now be plotted, see \cref{convex}. Keep in mind that the set $\Lambda_{\Delta}$ is a conservative estimate for the set of all $\lambda \in \wh{\Delta}$ for which $F_\lambda$ is convex over the simplex, i.e.,
$$
\Lambda_{\Delta}  \subseteq \{(\lambda_2, \lambda_3, \lambda_4) \in \wh{\Delta} : F_\lambda \text{ is convex on } \Delta^n \}.
$$
Hence, the region $\Lambda_{\Delta}$ shown in \cref{convex} should be thought of as pessimistic, and similarly for $\Lambda_{\square}$.
Furthermore, even if $F_\lambda$ is non-convex, one can still optimize (\ref{LMVSK}) and hope that the local optimum attained is sufficiently good.

The function of these sets is as follows. By optimizing  $F_\lambda$ for different $\lambda \in \Delta^4$, we recover local optimizers $w_\lambda$. If $\lambda \in \Lambda_{\Delta}$ then, by \cref{cor_convex}, we know that the optimizer $w_\lambda$ is globally optimal for problem (\ref{LMVSK}). If additionally we know that $\lambda > 0$, then by \cref{lem_neat} we know that $w_\lambda$ is a Pareto optimal point of problem (\ref{MVSK}).
Later in \cref{sec_num}, we will visualize the quality of solutions $w_\lambda$ by plotting objective values  $f_i(w_\lambda)$ against $\lambda \in \wh{\Delta}$, for $i \in [4]$. Hence, the sets $\Lambda_{\Delta}$, $\Lambda_{\square}$ are useful in showing where we certainly have Pareto optimality.

\begin{figure}[h]
\centering
\includegraphics[scale=0.75]{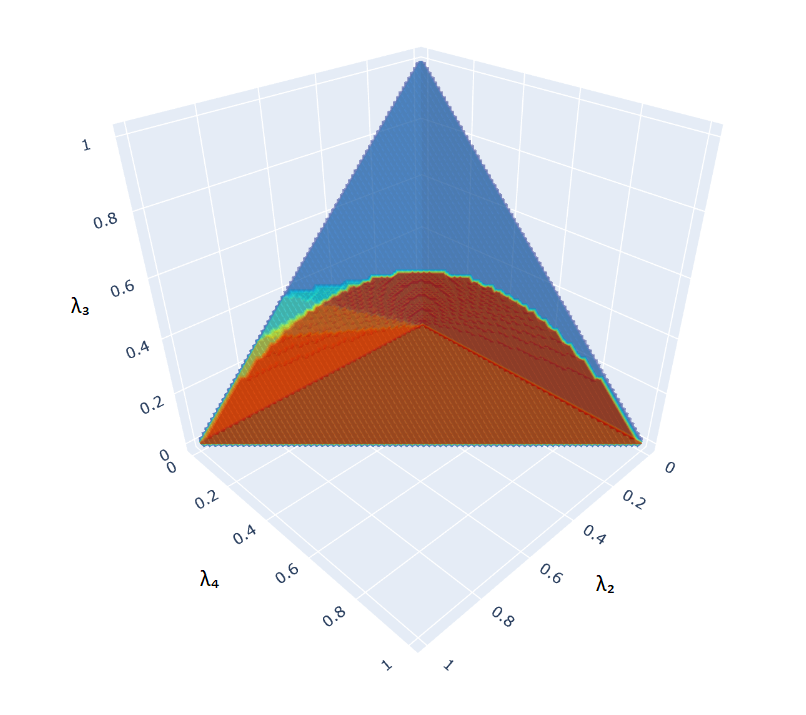}
\qquad
\subfloat[Subfigure 1 list of figures text][{\color{red}$\Lambda_{+}$}]{
\includegraphics[width=0.33\textwidth]{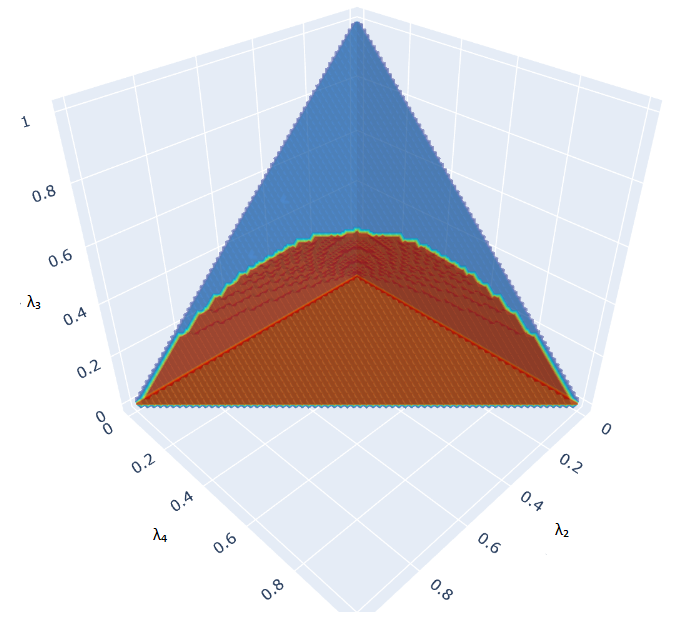}
\label{fig:subfig1}}
\subfloat[Subfigure 2 list of figures text][{\color{olive}$\Lambda_{\square}$}]{
\includegraphics[width=0.33\textwidth]{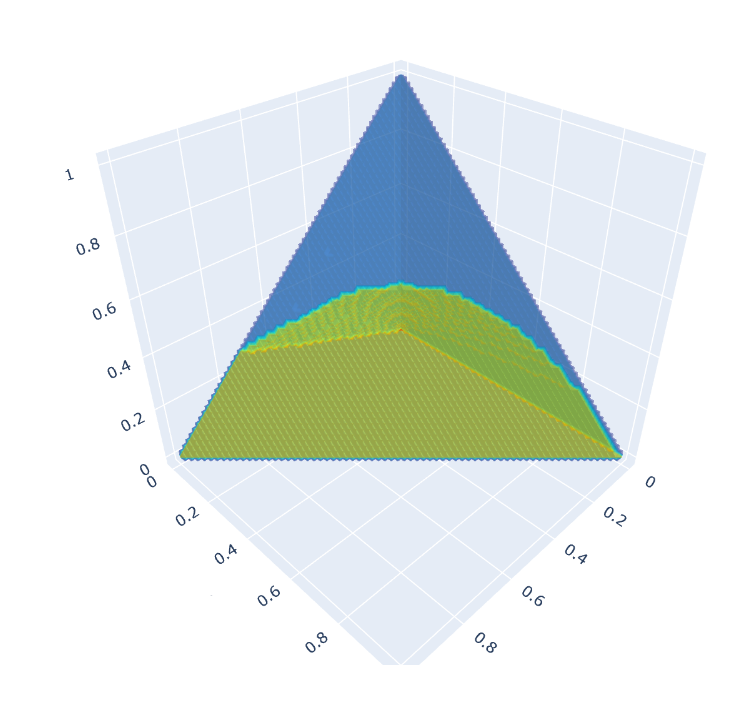}
\label{fig:subfig2}}
\subfloat[Subfigure 3 list of figures text][{\color{teal}$\Lambda_{\Delta}$}]{
\includegraphics[width=0.33\textwidth]{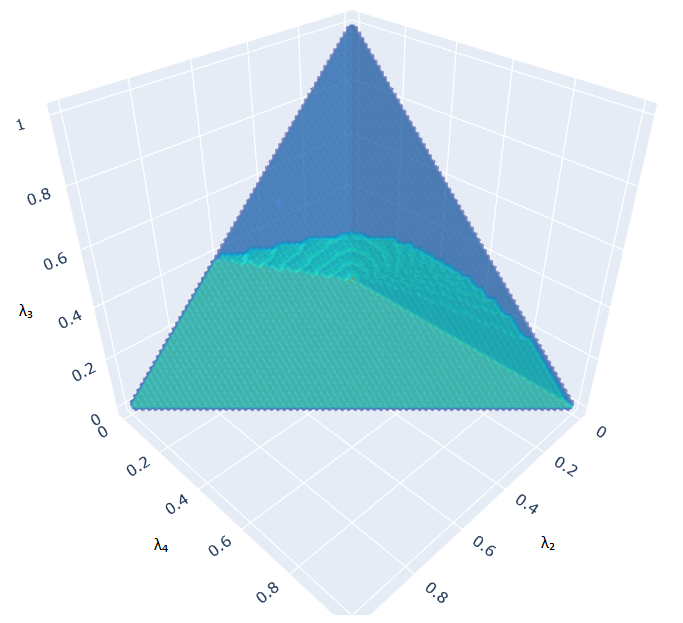}
\label{fig:subfig3}}
\caption{This plot shows the transparent three-dimensional hyper-parameter set  {\color{blue}$\wh{\Delta}$} in {\color{blue}blue} as viewed from the facet: $\{(\lambda_2, \lambda_3, \lambda_4) \geq 0 :  \lambda_2+ \lambda_3 + \lambda_4 = 1 \}  \subset \wh{\Delta}$. The different regions are distinguished by color. In particular,  {\color{red}$\Lambda_{+}$} is shown in {\color{red}red}, {\color{olive}$\Lambda_{\square}$} is shown in {\color{olive}green}, and {\color{teal}$\Lambda_{\Delta}$} is shown in {\color{teal}light-blue}. The domains $\Lambda_{\Delta}$ and $\Lambda_{\square}$ shown here were computed using $\Rtri=0.52$, $\Rsqr=0.87$, and $\Rtriu,~\Rsqru <0$.  
For this instance $\Lambda_{\square} \subseteq \Lambda_{\Delta}$.
The approximate relative volumes for the sub-domains are as follow: $\frac{\vol(\Lambda_{+})}{\vol(\wh{\Delta})}  \approx  0.59$, $\frac{\vol(\Lambda_{\square})}{\vol(\wh{\Delta})}  \approx  0.61$, and $\frac{\vol(\Lambda_{\Delta})}{\vol(\wh{\Delta})}  \approx  0.63$.}
\label{convex}
\end{figure}
\FloatBarrier

\subsection{Scalarized sparse MVSK} \label{sub_sec_spar_scal}
Analogous to the above discussion, one can associate a linear scalarization to the sparse MOOP in (\ref{SMVSK}) as follows:
\begin{equation} \label{SLMVSK}
\begin{split} 
F_{\lambda, {\cal C} }^* :=  & \min~~ F_\lambda(w) := -\lambda_1 f_1(w) + \lambda_2 f_2(w) -\lambda_3 f_3(w) + \lambda_4 f_4(w) \\
& s.t. ~~  w \in \Delta^n \\
&~~~~~ \prod_{i\in C}w_i = 0 ~\text{for}~C \in {\cal C}.
\end{split}
\end{equation}
We now show that optimization problems of the form (\ref{SLMVSK}) can be decomposed into a collection of several independent sub-problems of the form (\ref{LMVSK}). The motivation for doing this is that the sub-problems could possibly be solved independently using parallelization or other forms of distributed computing.  

For any set $U \subseteq [n]$ and vector $x \in \R^{U}$ denote by $x(0,U) \in \R^n$ the lifting of $x$ into $\R^n$, defined entrywise by
\begin{equation*}  
x(0,U)_i :=  \left\{
        \begin{array}{ll}
            x_i & \quad i \in U \\
            0   & \quad i \notin U
        \end{array}
    \right. ~~(i \in [n]).
\end{equation*}
Similarly, for $x\in \R^n$, let $x_{|_U} = (x_i)_{i \in U} \in \R^{U}$ denote the restriction of $x$ to $\R^{U}$. For a function $g:\R^n \to \R $ define the restricted function $g_{|_U}:\R^{U} \to \R,~x \mapsto g(x(0,U))$.

\begin{lemma} \label{spar_poly_lem}
Let $U_1,...,U_p \subseteq [n]$ be all the maximal subsets of $[n]$ not containing any set $C \in {\cal C}$. Then we have
$$
F_{\lambda,{\cal C}}^* = \wtl F_{\lambda,{\cal C}}^*,
$$
where
\begin{equation} \label{SLMVSK_alt}
\wtl F_{\lambda,{\cal C}}^* := \min_{\ell \in [p]} \min_{\wtl w \in \Delta^{U_\ell}} F_{\lambda|_{U_\ell}}(\wtl w).  \\
\end{equation}
\end{lemma}

\begin{proof}
\textbf{($F_{\lambda,{\cal C}}^* \geq \wtl F_{\lambda,{\cal C}}^*$)}
Any optimal solution $w$ of (\ref{SLMVSK}) must have its support contained in some $U_\ell$. Hence there is a $\wtl w \in \Delta^{U_\ell}$ such that $\wtl w(0,U_\ell) = w$ and $\wtl F_{\lambda,{\cal C}}^* \leq F_{\lambda|_{U_\ell}}(\wtl w) = F_\lambda(w) = F_{\lambda,{\cal C}}^*$.\\
\textbf{($F_{\lambda,{\cal C}}^* \leq \wtl F_{\lambda,{\cal C}}^*$)}
Let $\wtl w \in \Delta^{U_\ell}$ for some $\ell \in [p]$ be an optimizer of (\ref{SLMVSK_alt}). Then $\wtl F_{\lambda,{\cal C}}^* = F_{\lambda|_{U_\ell}}(\wtl w) = F_{\lambda}(\wtl w(0,U_\ell)) \geq  F_{\lambda,{\cal C}}^* $.
\end{proof}
The above result is reminiscent of Proposition 6 \cite{https://doi.org/10.48550/arxiv.2209.09573}. In [15], the concept of ideal-sparsity was introduced in the context of the generalized moment problem (GMP), when one restricts the support of the involved measure. This GMP with a single measure with restricted support can be shown to be equivalent to another GMP involving several measures, each having smaller   support than the measure in the original GMP. In \cref{spar_poly_lem}, we show that a polynomial optimization problem with support constraints is equivalent to optimizing over a set of smaller restricted polynomial optimization problems.
In both settings, the critical insight is that the restricted support constraint decomposes into a collection of smaller restrictions.

\subsubsection{Convexity of the scalarized sparse MVSK}

Similar to the dense case in \cref{convex_res}, the objective function in (\ref{SLMVSK}) is convex if $\lambda$ satisfies any of the conditions (i)-(iv) of \cref{pos_coef}.
The result of \cref{pos_coef} transfers to the sparse case because $\qua \geq 0$ on  $\Delta^n$ implies that $\qua \geq 0$ on 
\begin{equation} \label{def_spar_delta}
\Delta^n_{\cal C} := \{ w \in \Delta^n : \prod_{i\in C}w_i = 0 ~\text{for}~C \in {\cal C} \} \subseteq  \Delta^n.
\end{equation}
Hence \cref{pos_coef} and its consequences continue to hold in the sparse setting. Note that the domain $\Delta^n_{\cal C}$ is not convex, so the problem (\ref{SLMVSK}) is not convex.
 However, if $U_1,...,U_p \subseteq [n]$ denote all the maximal subsets of $[n]$ not containing any set $C \in {\cal C}$, then for any $\ell \in [p]$ the sub-problem
$$
\min_{\wtl w \in \Delta^{U_\ell}} F_{\lambda|_{U_\ell}}(\wtl w),
$$
does have a convex domain, i.e., $\Delta^{U_\ell}$. Moreover, on this sub-problem, \cref{pos_coef} can be adapted by using the following bounds
$$
\qu{R}_{ \Delta^{U_\ell}} := \min_{\wtl w \in \Delta^{U_\ell}}\inp{R,\wtl w(0,U_\ell)},~ \qo{R}_{\Delta^{U_\ell}} := \max_{\wtl w \in \Delta^{U_\ell}}\inp{R,w(0,U_\ell)}.
$$
Observe that $\Rtriu \leq  \qu{R}_{ \Delta^{U_\ell}} \leq \qo{R}_{\Delta^{U_\ell}} \leq \Rtri$ for all $\ell \in [p]$.
Furthermore, any $\lambda$ that satisfies at least one of the conditions (i)-(iv) of \cref{pos_coef} using the bounds $\Rtriu$ and $\Rtri$ will necessarily again satisfy one of the conditions using instead now the bounds $\qo{R}_{ \Delta^{U_\ell}}$ and $\qu{R}_{ \Delta^{U_\ell} }$, for any $\ell \in [p]$. Intuitively one can think of using these new bounds $\qu{R}_{ \Delta^{U_\ell}}$  and $\qo{R}_{\Delta^{U_\ell}}$ as relaxing the condition  $\qua \geq 0$ for all $w \in \Delta^n$ to the weaker condition $\qua \geq 0$ for all $w \in \Delta^n $ with $\supp(w) \subseteq U_\ell$. This mirrors the fact that there are potentially more hyper-parameters $\lambda \in \Delta^4$ for which $F_{\lambda|_{U_\ell}}$ is convex over $\Delta^{U_\ell}$ for each $\ell \in [p]$ than there are $\lambda \in \Delta^4$ for which $F_\lambda$ is convex over $\Delta^n$. 

The sparse problem (\ref{SLMVSK}) could have combinatorially many sub-problems to solve, but each sub-problem is smaller than the original problem and can be solved independently of the other sub-problems. If we set ${\cal C} $ to be the collection of all sets of size $k+1$, then there are ${n \choose k}$ sub-problems to solve, each involving $k$ variables.

\section{Numerical experiments} \label{sec_num}
In this section, we apply the theory from the preceding sections to real-world data.
We discuss the optimization algorithm \textit{FISTA}, by Beck and Teboulle \cite{doi:10.1137/080716542}, that we use to solve the scalarized problem (\ref{LMVSK}), and we motivate its use by listing some of FISTA's desirable properties.
We explain our methodology for acquiring a grid approximation for the Pareto set of MVSK problem (\ref{MVSK}).
Having obtained a set of optimizers of the scalarized problem (\ref{LMVSK}) for different choices of hyper-parameters, we compare and visualize the objective values of the MOOP (\ref{MVSK}) at said optimizers. We observe that some optimizers give a better overall balance among the four objectives.
This procedure is performed for the simplex and cube settings as well as their sparse analogs.

\subsection{Optimization algorithm for the scalarized problem}\label{sub_sec_fista}
\textit{Fast iterative shrinkage-thresholding algorithm} (FISTA), also known as \textit{fast proximal gradient method}, is a well-studied first-order iterative optimization algorithm first devised and analyzed by Beck and Teboulle \cite{doi:10.1137/080716542}.
Consider the scalar optimization problem (\ref{LMVSK}) and assume that $F_\lambda$ is convex. Then
under some mild smoothness assumptions, the details of which we omit to mention here for the sake of brevity, there is the following performance guarantee for the  k$^{\mathrm{th}}$ iteration of FISTA applied to (\ref{LMVSK}), see Theorem 10.34  \cite{doi:10.1137/1.9781611974997}:
$$
F_\lambda(x^k) - F_\lambda(x^*) \leq \frac{2L_F \|x^0-x^*\|^2 }{(k+1)^2}.
$$
Here $L_F > 0$ is the  Lipschitz constant of $F_\lambda$, $x^0$ is the initial point, $x^*$ is an optimizer, and  $x^k$ is the point obtained from FISTA at the k$^{\mathrm{th}}$ iteration. For all our application of FISTA we used $k=2000$ iterations.
For more details on FISTA, we refer to Chapter 10 of the monograph by Beck \cite{doi:10.1137/1.9781611974997}. We now proceed to mention some of the properties of FISTA that make it well suited for our problem.\\ 

Like many gradient descent algorithms, FISTA makes use of a projection operator in order to maintain the simplex (resp., cube) constraints. The operator that projects to the simplex defined by
$$
\proj_{\Delta^n}: \R^n \to \Delta^n,~x \mapsto \argmin_{y \in \Delta^n} \|x-y\|.
$$

If the nearest unconstrained optimizer lies outside of the simplex, then most gradient steps will leave the domain. Projecting back to the simplex results in a sparse vector, i.e., without full support. The sparsity seems to be due to the fact that projections are often on a face of the simplex. Hence, most optimizers obtained from FISTA will be sparse. We provide a histogram of the supports of optimizers from the set $\Wd$ (defined in \cref{sub_sub_optimizer_sets}) for our particular problem in \cref{supp_hist}. This sparsity does not occur in the case of the cube domain, i.e., the supports of $\Ws$ (defined in \cref{sub_sub_optimizer_sets}) are all full. One possible reason for this is that the unconstrained optimizer lies within the cube, and, as such, the projection operator does nothing. Note that the cube is full-dimensional in contradistinction to the simplex, which lies in the hyperplane $\{x \in \R^n : \sum_{i \in [n]}x_i = 1\}$.\\

\begin{figure}[ht]
	\includegraphics[scale=0.38]{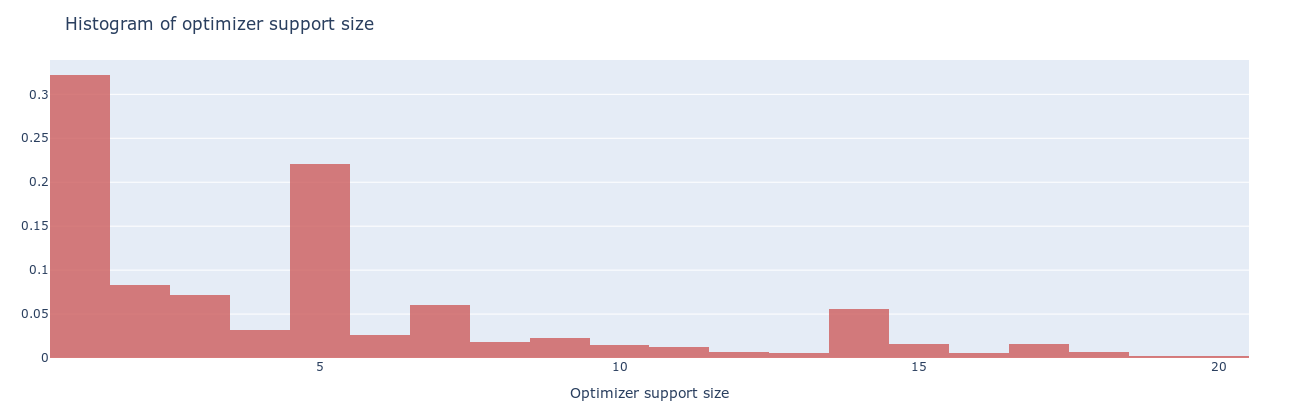}
	\caption{Normalized histogram of the support sizes $|\supp(w_\lambda)|$ of optimizers $w_\lambda \in \Wd$.} 
	\label{supp_hist}
\end{figure}
\FloatBarrier

FISTA is an iterative algorithm that starts from an initial guess $x^0$ and then incrementally improves a proposed optimizer until a certain number of iterations have been completed. In the convex problem, the algorithm will converge to the global optimum regardless of where one starts, but a closer start does imply faster convergence. Furthermore, if the problem is not convex, and one starts sufficiently close to the global optimizer, then one can be sure that FISTA will converge to the true optimum. An initial guess $x^0$ that is close to the global optimizer $x^*$ is called a \textit{warm start}. We now propose to use the optimizer from an already solved problem (\ref{LMVSK}), with a fixed $\lambda$, as a warm start for solving (\ref{LMVSK}) with a different hyper-parameter $\wh \lambda$. In other words, fix $\lambda \in \Delta^n$. If 
\begin{equation*}
\begin{split} 
w_\lambda \in ~&\argmin_{w \in \Delta^n} F_\lambda(w),  
 \end{split}
\end{equation*}
then take $x^0 = w_\lambda$ as a warm start for FISTA when solving
\begin{equation} \label{LMVSK_alt}
\begin{split} 
&\min_{w \in \Delta^n}  F_{\wh \lambda}(w).
 \end{split}
\end{equation}
Our intuition here is as follows: if $\lambda$ is close to $\wh \lambda$, then we expect $w_{\lambda}$ should be close to an optimizer $w_{\wh \lambda}$ of (\ref{LMVSK_alt}). Note that this is a heuristic and we provide no proof of the validity of this intuition. 
The same idea can be applied to computing sparse optimizers via FISTA. We elaborate more on this in the following sub-section.

\subsubsection{Optimization algorithm for the sparse scalarized problem}\label{sub_sub_sec_spar_fista}
We saw above that the optimizers of problem (\ref{LMVSK}) are sometimes sparse for the simplex setting, but not always.
So we propose a simple scheme for finding optimizers with support not exceeding some fixed integer $k \in \bb{N}$.
We do this starting from a set of possibly dense solutions $W$.
Let  $w_{\lambda} \in W$ be the optimizer of the (dense) problem (\ref{LMVSK}).
If $|\supp(w_{\lambda})|\leq k$ then we are done.
So suppose that $|\supp(w_{\lambda})|> k$. Keeping with the notation of \cref{sub_sec_spar_scal} we let ${\cal C} =  \{C \subseteq [n]:|C| > k \}$ and define the set
$$
{\cal U}^k :=  \{ U \subseteq [n]:|U| = k  \}
$$
of all maximal subsets of $[n]$ that do not contain any set $C \in {\cal C} $.
The sparse problem (\ref{SLMVSK}) can be rewritten as follows:
\begin{equation} \label{def_U_prob}
\begin{split}
\min_{\substack{w \in \Delta^n\\|\supp(w)| \leq k }} F_{\lambda}(w)  
= \min_{\substack{ U \in {\cal U}^k \\ \wtl w \in \Delta^U   }} F_{\lambda|_U}(\wtl w). 
\end{split}
\end{equation}
For a fixed $ U \in {\cal U}^k$  we can solve the sub-problem
\begin{equation} \label{def_sub_prob}
\begin{split} 
\min_{\wtl w \in \Delta^U  } F_{\lambda|_U}(\wtl w), 
\end{split}
\end{equation} 
using FISTA with $x^0 = \proj_{\Delta^{U}}(w_\lambda)$ as a warm start, where $w_\lambda$ is assumed to be an optimizer from the dense problem (\ref{LMVSK}) with the same hyper-parameter $\lambda$.
\paragraph{Removing sub-problems based on proximity to the dense optimizer.}
In order to not consider all ${n \choose k}$-many sets $U$ of $~{\cal U}^k$, we propose the following two heuristics to remove sets $U$ for which the resulting sub-problem (\ref{def_sub_prob}) could have a poor optimum value. The two heuristics we introduce can be used independently of each other. However, we will use them together in the sequence we introduce them.

The first heuristic consists of discarding all sets $U \in  {\cal U}^k$ that do not satisfy $U \subseteq \supp(w_{\lambda})$. Doing so yields only ${|\supp(w_{\lambda})| \choose k} \leq {n \choose k}$ sets to optimize over in (\ref{def_U_prob}). The second heuristic is to look at the elements $w_{\lambda,U}$ of the set
$$
W_{\lambda,k} := \{ w_{\lambda,U}:= \Big(\proj_{\Delta^U}(w_\lambda)\Big)(0,U) : U \in  {\cal U}^k,~ U \subseteq \supp(w_{\lambda}) \} \subseteq \R^n,
$$
obtained by projecting $w_{\lambda}$ onto $\Delta^U \subseteq \R^U$ and then lifting the projection to a vector in $\R^n$ by padding entries not supported by $U$ with zeros, for all appropriate sets $U$. To use the second heuristic independently of the first simply drop the $ U \subseteq \supp(w_{\lambda})$ constraint in the definition of $W_{\lambda,k}$. We can then choose to solve problem (\ref{def_sub_prob}) only over sets $U$ for which $w_{\lambda,U}$ is close to $w_\lambda$ in the Euclidean norm. For our implementation we take the sets $U$ corresponding to the $n$ closest $w_{\lambda,U}$ to $w_\lambda$. Though we provide no guarantee that choosing a $U \subseteq [n]$ such that $ w_{\lambda,U}$ is closest to $w_{\lambda}$ would result in an optimum value of problem (\ref{def_sub_prob}) being any better than another choice of $U$, we still find that this a helpful heuristic for removing poor choices of $U$.

\subsubsection{The set of obtained optimizers} \label{sub_sub_optimizer_sets}
Whether $F_\lambda$ is convex or not, we can apply FISTA to obtain at least a local optimizer $w_\lambda$ for problem (\ref{LMVSK}). Construct the following sets of optimizers, obtained by applying FISTA to various scalarizations:
\begin{equation} \label{W_delta}
\begin{split} 
&W_{\Delta} :=\{ w_\lambda \in \argmin_{\mathrm{FISTA}~w\in \Delta^n}F_\lambda(w) : \lambda \in \Delta^4 \}, \\
&W_{\square} :=  \{ w_\lambda \in \argmin_{\mathrm{FISTA}~w\in [-1,1]^n}F_\lambda(w) : \lambda \in \Delta^4  \}.
\end{split} 
\end{equation}
Here, $\argmin_{\mathrm{FISTA}}$ denotes the local minimizers obtained via the algorithm FISTA, not to be confused with the true (unknown) global minimizers.
In \cref{par_opt_simp} we will visualize the values of the objectives $f_1(w),f_2(w),f_3(w)$, and $f_4(w)$ for $w \in W_{\Delta}$ (resp., $w \in W_{\square}$) using colors.
Similarly, we construct the sets of sparse FISTA local optimizers
\begin{equation*}
\begin{split} 
&W_{\Delta,k} :=  \{ w_\lambda \in\argmin_{\mathrm{FISTA}~w\in \Delta^n,  ~|\supp(w)|\leq k~}F_\lambda(w) : \lambda \in \Delta^4 \}, \\
&W_{\square,k} := \{ w_\lambda \in \argmin_{\mathrm{FISTA}~w\in [-1,1]^n,  ~|\supp(w)|\leq k~}F_\lambda(w) : \lambda \in \Delta^4  \}, 
\end{split} 
\end{equation*}
obtained by following the procedure described in \cref{sub_sub_sec_spar_fista}.
As mentioned before, projecting onto the simplex often produces a sparse vector. Hence, it makes sense to use $w \in W_{\Delta}$ as a starting point for computing $W_{\Delta,k}$, as many of the vectors of $W_{\Delta}$ may already be sparse enough. Regardless of whether the elements of $W_{\Delta}$ (resp., $W_{\square}$) are sparse we can use the ideas of \cref{sub_sub_sec_spar_fista} to prune computations and generate warm starts for the problems associated with $W_{\Delta,k}$ (resp., $W_{\square,k}$).

\subsection{Defining objective functions from empirical data} \label{sub_sec_theo_emp}
For the sake of generality, we have worked with a vector-valued random variable $R$ (resp.,  $\wtl R$) taking values in $\R^n$. Practically, the data will arise from a table of results, taking the form of an $n \times m$ matrix $\wtl T \in \R^{n \times m}$, where $m$ is the number of outcomes observed over time. We introduce new notation for the empirical data $\wtl T$ and the subsequent derived quantities. The entry $\wtl T_{i,j}$ (resp.,  $T_{i,j}$) is interpreted as the \emph{empirical (resp., centralized) relative returns} of asset $i$ at time $j$. In this context, the expectation is taken over the outcomes. The mean becomes the \emph{empirical mean}, i.e.,
$$
M =   \Big(  \frac{1}{m}\sum_{p \in [m]} \wtl T_{i,p}  \Big)_{i \in [n]} \in \R^n.
$$
Hence, the \emph{empirical centralized relative returns} is defined for each $i \in [n]$ and $p \in [m]$  by 
$$
T_{i,p} :=  \wtl T_{i,p} - M_i.
$$
Similar to the mean, the formulation of the other empirical moments is as follows:
\begin{equation}\label{mom_funs_emp} 
\begin{split} 
&V = \Big(\frac{1}{m-1}\sum_{p,q \in [m]} T_{i,p} T_{j,q}  \Big)_{i,j \in [n]}, \\
&S = \Big(\frac{1}{m}\sum_{p,q,r \in [m]} T_{i,p} T_{j,q} T_{k,r}  \Big)_{(i,j) \in ([n] \times [n]),~k \in [n]}, \\
&K = \Big(\frac{1}{m}\sum_{p,q,r,s \in [m]} T_{i,p} T_{j,q} T_{k,r} T_{\ell,s}  \Big)_{ (i,j), (k,\ell) \in ([n] \times [n])}. \\
\end{split}
\end{equation}
Observe that we use the unbiased estimator of the variance in (\ref{mom_funs_emp}); for a general reference on statistical estimators, we refer to \cite{wackerly2014mathematical}.
The objective functions $f_1$, $f_2$, $f_3$, and $f_4$ defined in (\ref{mean_fun}) and (\ref{mom_funs}) can henceforth be redefined in terms of the above $M$, $V$, $S$, or $K$. \\

\noindent Using $T \in \R^{n \times m}$, the bounds $\Rtri$ and $\Rtriu$ in \cref{pos_coef} now become
$$
\Rtri = \max_{i\in [n], p \in [m]}T_{i,p},~\Rtriu = \min_{i\in [n], p \in [m]}T_{i,p}.
$$
For the cube the bounds are $\Rsqr = \max_{ p \in [m]} \sum_{i \in [n]}|T_{i,p}|$ and $\Rsqru = -\Rsqr$. The sparse analogs $\qu{R}_{ \Delta^{U}}$, $\qo{R}_{\Delta^{U}}$, $\qu{R}_{ \square^{U}}$, and $\qo{R}_{\square^{U}}$  are defined, mutatis mutandis, in the same manner. The bounds we gave in \cref{convex} are also used for all computation we show. We only compute and use the dense bounds ($\Rtri$, $\Rtriu$, $\Rsqru$, and $\Rsqr$), even for the sparse settings. \\

\noindent In the next section we sub-sample the sets $W_{\Delta}$, $W_{\square}$, $W_{\Delta,5}$, and $W_{\square,5}$, described in the preceding section.
Our empirical data $\wtl T$ will be a selection of stocks from the well-known \emph{Standard and Poor's 500} (S{\&}P500) stock market index, see \cite{editorsofencyclopaediabritannica5002023}. We will consider $n=20$ stocks, each measured in increments of a day over a timespan of $m=500$ days starting in January 1990. We have chosen this dataset because it is well known and publicly available. However, everything we describe in this paper could also be applied to any other time series data of asset prices. For the reader's convenience we list some papers \cite{Mhiri2010InternationalPO,Aracolu2011MeanVarianceSkewnessKurtosisAT} that investigate the MVSK model on markets different from the S{\&}P500. Using $\wtl T$ we can generate $M$, $V$, $S$, and $K$ as described above. From here we can define problem (\ref{LMVSK}) and its sparse analog (\ref{SLMVSK}). Solving these problems, using the procedure described in \cref{sub_sec_fista}, we obtain elements from the sets $W_{\Delta}$, $W_{\square}$, $W_{\Delta,5}$, and $W_{\square,5}$.  \\

\subsection{A grid approximation of the Pareto set}\label{sub_sec_grid}
Recall that the ultimate goal is to obtain Pareto optimizers of the MVSK problem (\ref{MVSK}).
Via \cref{lem_neat}, solving  the scalarization (\ref{LMVSK}) for $\lambda > 0$ gives a Pareto optimizer of (\ref{MVSK}).
However, we can still recover an optimizer from  solving  the scalarization (\ref{LMVSK}) for $\lambda \geq 0$, we simply have no guarantee of them being Pareto optimal in (\ref{MVSK}).
Because $\Delta^4$ contains uncountably many elements we resort to sub-sampling $\Delta^4$ with a uniform mesh. Fix $s \in \bb{N}$, and consider the following sets:
\begin{equation*}  
\begin{split} 
&\Delta^4_{[s]} := \{\lambda : \lambda \in \{0,\frac{1}{s},\frac{2}{s},...,1\}^4 \cap \Delta^4\}, \\
&\wh \Delta_{[s]}  := \{(\lambda_2, \lambda_3, \lambda_4) : (1-(\lambda_2+ \lambda_3+ \lambda_4),\lambda_2, \lambda_3, \lambda_4) \in  \Delta^4_{[s]}   \}  \subseteq  \R^3, 
\end{split} 
\end{equation*}  
that are clearly in bijection.
For our computations we take $s =40$, resulting in $|\wh \Delta_{[40]}|= 11480$ choices of hyper-parameter $\lambda$ to consider.
For each $\lambda \in \Delta^4_{[40]}$, we solve the associated scalarization (\ref{LMVSK}) using FISTA to obtain a set of local optimizers $w_\lambda$, denoted by
$$
\Wd \subseteq \{ w_\lambda \in \argmin_{\mathrm{FISTA}~w\in \Delta^n}F_\lambda(w) : \lambda \in \Delta^4_{[40]} \} \subseteq W_{\Delta}.
$$
Observe that the set $\Wd$ is not necessarily contained in the Pareto front, but the following subset is:
$$
 \{ w_\lambda \in \Wd  : \lambda \in \Lambda_{\Delta},~ \lambda >0  \}.
$$
Here we use the claims from \cref{convex_res} and \cref{lem_neat} that if $\lambda \in \Lambda_{\Delta}$ and $\lambda > 0$ then $w_\lambda$ is a Pareto optimizer of problem (\ref{MVSK}). The reason we consider the bigger set $\Wd$ is that we get a more complete picture, see the figures of \cref{sub_sec_num_res}. Although some points of $\Wd$ are not guaranteed to be a Pareto optimizer of (\ref{MVSK}), they are nonetheless quite comparable to the points that are Pareto optimal for (\ref{MVSK}). We illustrate this claim with visualization in the subsequent sub-sections of \cref{sub_sec_num_res}. 

In order to compare points $w \in \Wd$, we rank them in terms of their values for the objective functions $f_1$, $f_2$, $f_3$, and $f_4$ in (\ref{MVSK}).
For each $w \in \Wd$ we compute the values $f_1(w)$, $f_2(w)$, $f_3(w)$, and $f_4(w)$.
For the sake of clarity, since there is a scale difference between the different functions, we linearly rescale the values to be in the interval $[0,1]$.
Formally, for each $i \in [4]$ define
$$
\Fid : = \{ f_{i}^{[40]}(w): w  \in \Wd \},
$$
to be the set of linearly scaled values  $f_{i}(w)$ for $w \in \Wd$,  where
\begin{equation} \label{scaling} 
f_{i}^{[40]}(w) :=
\left\{
        \begin{array}{ll}
            \frac{f_{i}(w) -  f_{i}^{\mathrm{min},[40]}}{f_{i}^{\mathrm{max},[40]} - f_{i}^{\mathrm{min},[40]}} & \quad i =1 \text{ or }3\\
            1 - \frac{f_{i}(w) -  f_{i}^{\mathrm{min},[40]}}{f_{i}^{\mathrm{max},[40]} - f_{i}^{\mathrm{min},[40]}} & \quad i=2 \text{ or }4\\
        \end{array},
    \right.
\end{equation}
with 
$$
f_{i}^{\mathrm{max},[40]} := \max_{w \in \Wd } f_{i}(w),~f_{i}^{\mathrm{min},[40]} :=  \min_{w \in \Wd } f_{i}(w).
$$
Hence, for any $i\in [4]$, the set $\Fid$ is contained in the unit interval $[0,1]$, with ``less desirable" values close to zero and ``more desirable" values close to one.
Note that the scaling $f_{i}^{[40]}$ considers the fact that we want to maximize $f_1$ and $f_3$, and to minimize $f_2$ and $f_4$.
Hence, the set $\Fid$ gives us a way to compare the performance of each portfolio $w \in \Wd$ with respect to the objective function $f_i$, for all $i \in [4]$. For each $i\in [4]$, we plot $\Fid$ (in color) against $\wh \Delta_{[40]}$ (in $\R^3$), see \cref{simp_fig}.

In order to aggregate the quality of an optimizer $w \in \Wd$ over all of the objectives $f_1$, $f_2$, $f_3$, and $f_4$, we propose looking at the value
$$
f^{[40]}(w) := \sum_{i \in [4]} f_{i}^{[40]}(w) \in [0,4].
$$
The intuition behind this value is that if $w \in \Wd$ has a value $f^{[40]}(w)$ close to four, then it does well among many of the objectives and is hence a superior choice to another solution $v \in \Wd$ for which $f_{i}^{[40]}(v) \geq f_{i}^{[40]}(w)$ for some $i \in [4]$ but $f^{[40]}(v) < f^{[40]}(w)$. We refer to the following set
$$
\Wds{\eta} : = \{  w  \in \Wd: f^{[40]}(w) \geq (1-\eta) \cdot \Big( \max_{w \in \Wd} f^{[40]}(w)  \Big)  \},
$$
where $\eta \in (0,1) $, as the set of portfolios with $\eta$\emph{-superior trade-off}, and we define the set of associated \emph{scores} 
$$
\Fstar{\eta} : = \{ f^{[40]}(w) : w \in \Wds{\eta}  \}.
$$ 
We plot $\Fstar{0.01}$ in color against $\wh \Delta_{[40]}$ in \cref{F_Pareto_simp_sum}. Our plots should not be compared to figures as those in \cite{maringerGlobalOptimizationHigher2009} where three of four objective are plotted against each other with two independent and the third dependent. We give a separate plot for each objective and we scale for comprehensibility.  \\

\paragraph{Handling the cube and sparse cases.}
Above, we have described the process for the simplex ($w \in \Delta^n$), but the treatment is analogous for the cube domain ($w \in [-1,1]^n$) and sparse domains ($w \in \Delta^n$, $|\supp(w)|\leq k$) and ($w \in [-1,1]^n$, $|\supp(w)| \leq k$ ). 
Notation-wise, the sets $\Ws$, $\Fs{i}~(i\in[4])$, $\Wss{\eta} $ and $\Fstars{\eta}$ are all defined analogously to the simplex case, now using the domain $w \in [-1,1]^n$ instead of $w \in \Delta^n$.
Similarly the sparse simplex sets are denoted by $\Wdk$, $\Fdk{i}~(i\in[4])$, $\Wdks{\eta}$ and $\Fkstar{\eta}$, where $k\in \bb{N}$ is an upper bound on the support size of the elements as described in \cref{sub_sub_sec_spar_fista}. The sparse cube sets, denoted $\Wsk$, $\Fsk{i}~(i\in[4])$, $\Wsks{\eta}$ and $\Fsstar{\eta}$, are defined, mutatis mutandis, in the same manner.

\subsection{Numerical results} \label{sub_sec_num_res}
This final subsection is the culmination of the preceding subsections.
For the S{\&}P500 data considered at the end of \cref{sub_sec_theo_emp} we compute $\Wd$, $\Fid (i \in [4])$, $\Wds{0.01}$, and $\Fstar{0.01}$.
For each $i \in [4]$ we plot $\Fid$ (in color) against the hyper-parameter set $\Delta_{[40]} \subseteq \R^3$.
Doing so, we observe how each portfolio $w_\lambda \in \Wd$ makes a trade-off between the objectives $f_1$, $f_2$, $f_3$, and $f_4$.
Which of the objectives are favoured by $w_\lambda$ is influenced by the choice of $\lambda$.
For example, for $\lambda_1 = 1-(\lambda_2+\lambda_3+\lambda_4)  \geq 0.4$ the portfolios  $w_\lambda$ tend to have values $f_{1}^{[40]}(w_\lambda)$ close to one, see \cref{simp_subfig1}.
Observations like these are useful to investors who can now visually navigate the $\Fid~(i \in [4])$ in \cref{simp_fig} to find a portfolio $w_\lambda \in \Wd$ that matches their risk preferences.

To see which $\lambda  \in \Delta_{[40]}$ correspond to portfolios $w_\lambda$ with a good balance of all four objectives we plot $\Fstar{0.01}$ (in color) against $\Delta_{[40]}$ (resp., $\Delta_{[40]}\cap \Lambda_{\Delta} $) in \cref{F_Pareto_simp_sum}. Note that the hyper-parameters $\lambda  \in \Delta_{[40]}$ for which $w_\lambda \in  \Wd \setminus\Wds{0.01}$ are not displayed so as not to clutter the plot.

Above, we explained the process for the (dense) simplex setting, but the same treatment applies to the cube and sparse settings, resulting in analogous figures and similar observations.

\subsubsection{Numerical results in the simplex setting: $w \in \Delta^n$} \label{par_opt_simp}
\begin{figure}[!htbp]
\centering
\subfloat[Subfigure 1 list of figures text][ $f_{1}^{[40]}(w_\lambda) \in \Fd{1}$ vs. $(\lambda_2,\lambda_3,\lambda_4) \in \wh \Delta_{[40]}$]{
\includegraphics[width=0.42\textwidth]{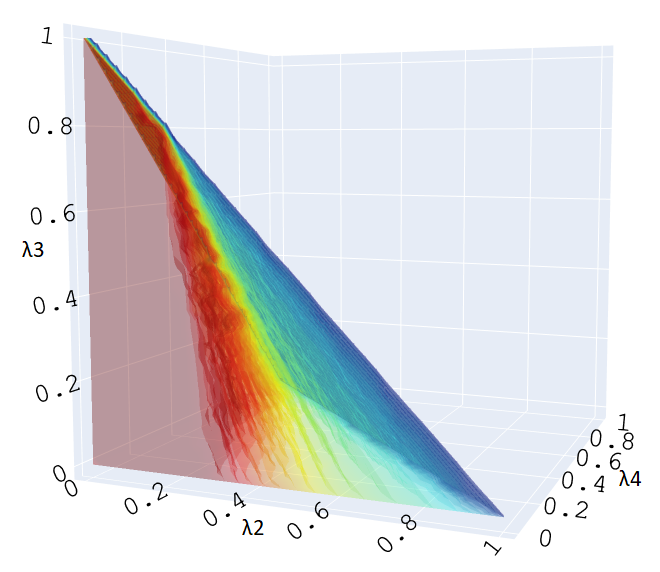}
\label{simp_subfig1}}
\subfloat[Subfigure 2 list of figures text][$\Fd{2}$ vs. $\wh \Delta_{[40]}$]{
\includegraphics[width=0.42\textwidth]{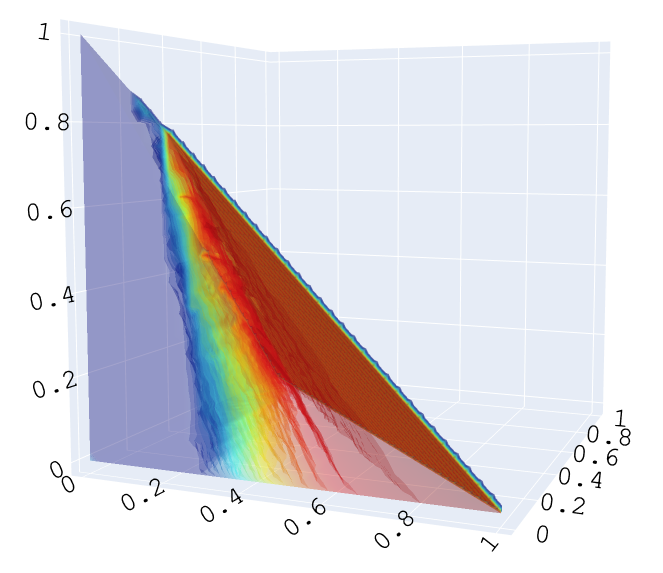}
\label{simp_subfig2}}
\qquad
\subfloat[Subfigure 3 list of figures text][$\Fd{3}$ vs. $\wh \Delta_{[40]}$]{
\includegraphics[width=0.42\textwidth]{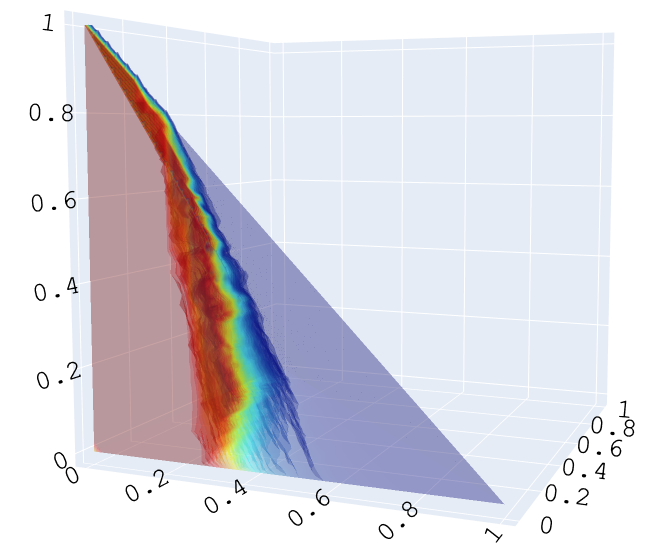}
\label{simp_subfig3}}
\subfloat[Subfigure 4 list of figures text][$\Fd{4}$ vs. $\wh \Delta_{[40]}$]{
\includegraphics[width=0.42\textwidth]{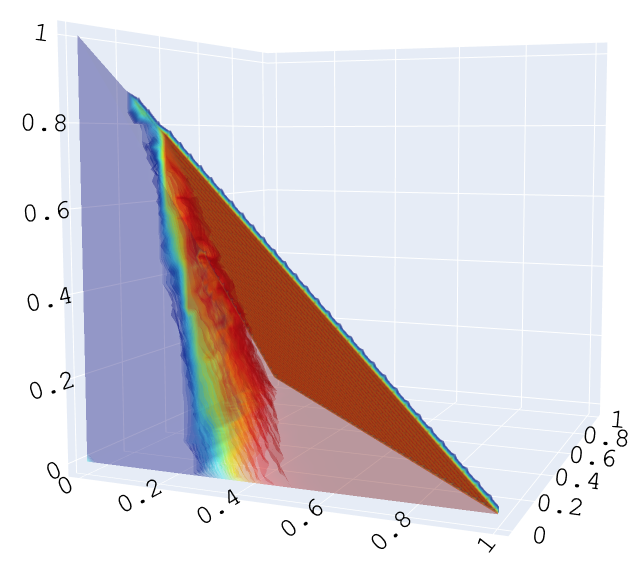}
\label{simp_subfig4}}
\qquad
\centering
\includegraphics[width=0.5\textwidth]{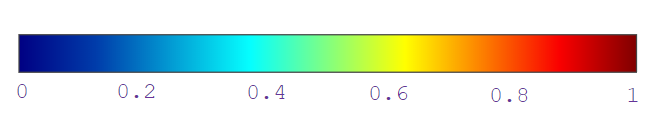}
\caption{This figure shows the transparent three-dimensional plots of $\Fid$ ($i \in [4]$) (in color) versus $(\lambda_2,\lambda_3,\lambda_4) \in \wh \Delta_{[40]}$, viewed from the facet: $\{(\lambda_2, \lambda_3, \lambda_4)  \in \wh \Delta:  \lambda_4 = 0 \}$. For every $i \in [4]$, every point $(\lambda_2, \lambda_3, \lambda_4) \in  \wh \Delta_{[40]}$ is assigned a color $f_{i}^{[40]}(w_\lambda) \in [0,1]$, where $w_\lambda \in \Wd$. Hence, {\color{red}red} regions correspond to better values while {\color{blue}blue} regions correspond to worse values.}
\label{simp_fig}
\end{figure}

\begin{figure}[!htbp]
\centering 
    \subfloat[Subfigure 1 list of figures text][$f^{[40]}(w_\lambda)\in \Fstar{0.01}$ vs.\\ $(\lambda_2,\lambda_3,\lambda_4) \in \wh \Delta_{[40]}$ such that $w_\lambda \in \Wds{0.01}$]{
	\includegraphics[scale=0.5]{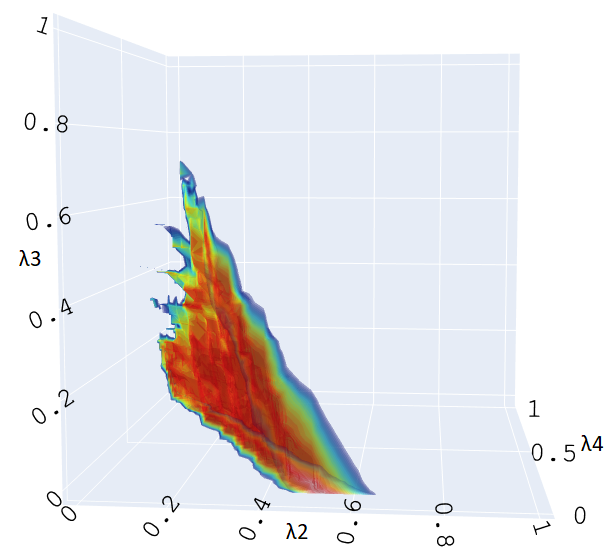}
	\label{simp_superior_subfig1}}
	\subfloat[Subfigure 2 list of figures text][$\Fstar{0.01}$ vs. $\wh \Delta_{[40]} \cap \Lambda_{\Delta}$]{
	\includegraphics[scale=0.5]{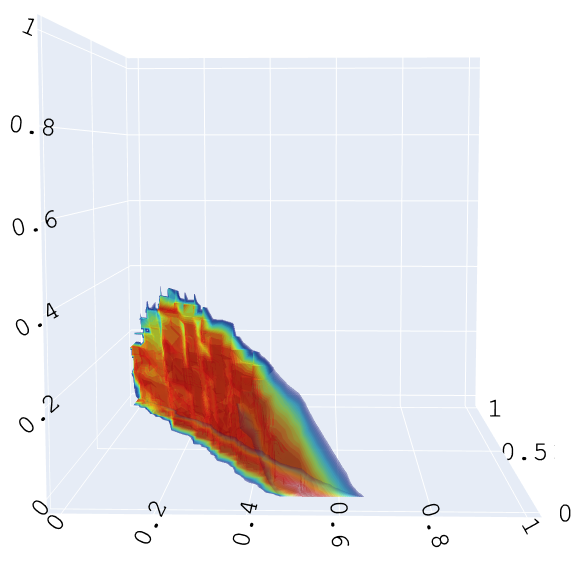}
	\label{simp_superior_subfig2}}
	\qquad
	\centering
	\includegraphics[width=0.5\textwidth]{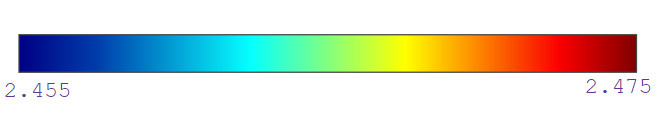}
	
	\caption{ This figure shows the transparent three-dimensional plot of $f^{[40]}(w_\lambda) \in \Fstar{0.01}$ in color versus $(\lambda_2,\lambda_3,\lambda_4) \in  \wh \Delta_{[40]}$ such that $w_\lambda \in \Wds{0.01}$, viewed from the facet: $\{(\lambda_2, \lambda_3, \lambda_4)  \in \wh \Delta:  \lambda_4 = 0 \}$. In particular, every point $(\lambda_2, \lambda_3, \lambda_4) \in  \wh \Delta_{[40]}$ is assigned a color $f^{[40]}(w_\lambda) \in [0,4]$. The values $ \Fstar{0.01}$ range from $0.99 \cdot 2.475$ to $\max_{w_\lambda \in \Wds{0.01} }f^{[40]}(w_\lambda) \approx 2.475$, which is indicated by the color bar. Again, {\color{red}red} regions correspond to better values while {\color{blue}blue} regions correspond to worse values.} 
\label{F_Pareto_simp_sum}
\end{figure}

\begin{table}[!htbp]
\caption{Selected results for $\lambda \in \wh \Delta_\Delta^{[40],0.01} \subseteq \R^4$.}
  \begin{center}
    \label{tab:table1}
    \begin{tabular}{c|c|c|c|c|c}
    $\lambda$& $f_{1}^{[40]}(w_\lambda)$&  $f_{2}^{[40]}(w_\lambda)$&  $f_{3}^{[40]}(w_\lambda)$&  $f_{4}^{[40]}(w_\lambda)$& $|\supp(w_\lambda)|$ \\ \hline
  
 [0.154,  0.256,  0.077,  0.513 ]& 0.623&  0.81 &  0.058&  0.978&  5 \\  \hline
 [0.026,  0.077,  0.256,  0.641 ]& 0.601&  0.825&  0.05 &  0.98 & 10\\ \hline
 [0.231,  0.41 ,  0.308,  0.051 ]& 0.581&  0.854&  0.034&  0.989&  5\\ \hline
 [0.462,  0.513,  0.026,  0.0   ]& 0.691&  0.724&  0.118&  0.942&  5\\ \hline
 [0.051,  0.051,  0.205,  0.692 ]& 0.677&  0.741&  0.104&  0.95 &  5\\ \hline
 [0.179,  0.359,  0.308,  0.154 ]& 0.562&  0.872&  0.026&  0.992&  5\\ \hline
 [0.462,  0.41 ,  0.128,  0.0   ]& 0.774&  0.586&  0.241&  0.85 &   3\\ \hline
 [0.282,  0.333,  0.385,  0.0   ]& 0.752&  0.625&  0.203&  0.881&   5\\ \hline
 [0.154,  0.231,  0.487,  0.128 ]& 0.676&  0.742&  0.104&  0.95 &   5\\ \hline
 [0.256,  0.256,  0.077,  0.41  ]& 0.715&  0.686&  0.148&  0.922&   5  \\ \hline
    \end{tabular}
  \end{center}
\end{table}

In \cref{simp_fig}, regions where the objectives $f_1$ and $f_3$ perform well (are red) overlap heavily, see \cref{simp_subfig1} and \cref{simp_subfig3}. Furthermore, these regions overlap with the regions where the objectives $f_2$ and $f_4$ do poorly (are blue), namely the rear slice of the simplex where either $\lambda_2$ or $\lambda_4$ is small, see \cref{simp_subfig2} and \cref{simp_subfig4}. The central wedge, $(\lambda_2,\lambda_3,\lambda_4) \in \wh \Delta_{[40]}$ such that $w_\lambda \in \Wds{0.01}$, where the objectives seem to balance out is shown in \cref{simp_superior_subfig1} along with the same wedge restricted to $\Lambda_{\Delta}$, shown in \cref{simp_superior_subfig2}. 

Recall from definition (\ref{conv_doms_3}) that $\Lambda_{\Delta}$ is a set of hyper-parameters $\lambda$ for which $F_\lambda$ is convex over the simplex $\Delta^n$.  Further recall that FISTA converges to a global minimizer when applied to a convex problem. With this in mind one would expect the quality of optimizers produced by FISTA to decline as $\lambda$ leaves $\Lambda_{\Delta}$ and $F_\lambda$ (possibly) ceases to be convex. However, this is not apparent from our plots. Observe how there does not seem to be a change in color in the plots of \cref{simp_fig,F_Pareto_simp_sum} as the hyper-parameters $\lambda$ move out of the region $\Lambda_{\Delta}$. This hints at the possibility that the local optima obtained by FISTA for hyper-parameters in $\wh{\Delta} \setminus \Lambda_{\Delta}$ are not much worse than the global optima.

Lastly, observe how the set 
$$
\wh\Delta_\Delta^{[40],0.01} := \{ \lambda \in \wh{\Delta}_{[40]} : w_\lambda \in \Wds{0.01} \}
$$
of hyper-parameters corresponding to solutions of superior trade-off overlap with the respective sets $\{\lambda \in \wh{\Delta} :  \lambda > 0\}$ and $\Lambda_{\Delta}$, see \cref{fig:subfig3} and \cref{simp_superior_subfig2}. 
In fact, the approximate volumes of these two sets and their intersection relative to $\wh\Delta_\Delta^{[40],0.01}$ are as follows:
\begin{equation*}  
\begin{split} 
&\frac{\vol( \{\lambda \in \wh\Delta_\Delta^{[40],0.01} :  \lambda > 0\}\cap \Lambda_{\Delta})}{\vol(\wh\Delta_\Delta^{[40],0.01})} \approx 0.77,\\
&\frac{\vol( \{\lambda \in \wh\Delta_\Delta^{[40],0.01} :  \lambda > 0\})}{\vol(\wh\Delta_\Delta^{[40],0.01})} \approx 0.83,\\
&\frac{\vol(\wh\Delta_\Delta^{[40],0.01}\cap \Lambda_{\Delta}) }{\vol(\wh\Delta_\Delta^{[40],0.01})}  \approx 0.90.
\end{split}
\end{equation*} 
Hence,  by virtue of \cref{lem_neat} and \cref{cor_convex}, we have that approximately $77\%$ of the optimizers with superior trade-off in $\Wds{0.01}$ are guaranteed to be Pareto optimizers of the MVSK problem (\ref{MVSK}).

For concreteness we show in \cref{tab:table1} the numerical values $f_{i}^{[40]}(w_\lambda) ~(i \in [4])$ for ten randomly selected hyper-parameters $\lambda \in \wh\Delta_\Delta^{[40],0.01}$.
We make the following observations. First, the skewness, i.e., $f_{3}^{[40]}(w_\lambda)$, seems to be the weakest performing objective relative to the others. Second, variance and kurtosis, i.e., $f_{2}^{[40]}(w_\lambda)$ and $f_{4}^{[40]}(w_\lambda)$, seem to be positively correlated. Third, the associated portfolios $w_\lambda$ are all sparse with at least half of their entries zero. Eight of the ten portfolios in  \cref{tab:table1} have support size $5$, this corroborates the data in the histogram shown in \cref{supp_hist}.

In the literature computational results are often represented in tabular from as we did in \cref{tab:table1}, see for example\cite{Lai2006MeanVarianceSkewnessKurtosisbasedPO,RePEc:com:wpaper:021}. Presenting results in this way for a large selection of hyper-parameters soon becomes cumbersome, especially in our case where we have $|\wh \Delta_{[40]}|= 11480$ (recall \cref{sub_sec_grid}). Moreover the overall patterns are often obscured by the detail of each specific entry. In contradistinction, by representing the results as we did in \cref{simp_fig} and \cref{F_Pareto_simp_sum} we see larger trends across the various choices of hyper-parameters $\lambda$. Hence, via the grid sampling approach of \cref{sub_sec_grid} and the visualizations of this section we believe we get a better overall understanding of the relationship between $\lambda$, $w_\lambda$, and the objective values $f_i(w_\lambda)~(i \in [4])$ than by simply looking at a few specific values of $\lambda$. \\

We now proceed with the other settings (sparse simplex, cube, and sparse cube), which follow in a similar manner.
Because the subsequent sub-sections follow the same format as this one, we omit describing the figures and focus instead on the differences and new observations.


\pagebreak
\subsubsection{Numerical results in the sparse simplex setting: $w \in \Delta^n$, $|\supp(w)|\leq 5 $} \label{par_opt_simp_spar}
\begin{figure}[!htbp]
\centering
\subfloat[Subfigure 1 list of figures text][$\Fdfive{1}$ vs. $ \wh \Delta_{[40]}$]{
\includegraphics[width=0.42\textwidth]{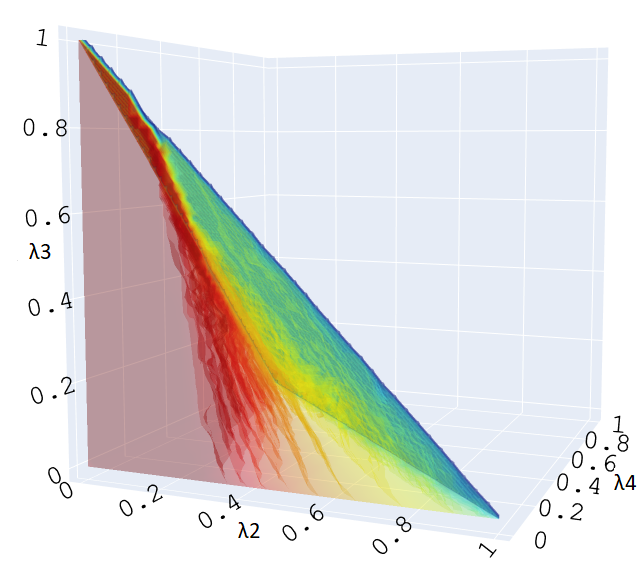}
\label{simp_spar_subfig1}}
\subfloat[Subfigure 2 list of figures text][$\Fdfive{2}$ vs. $\wh \Delta_{[40]}$]{
\includegraphics[width=0.42\textwidth]{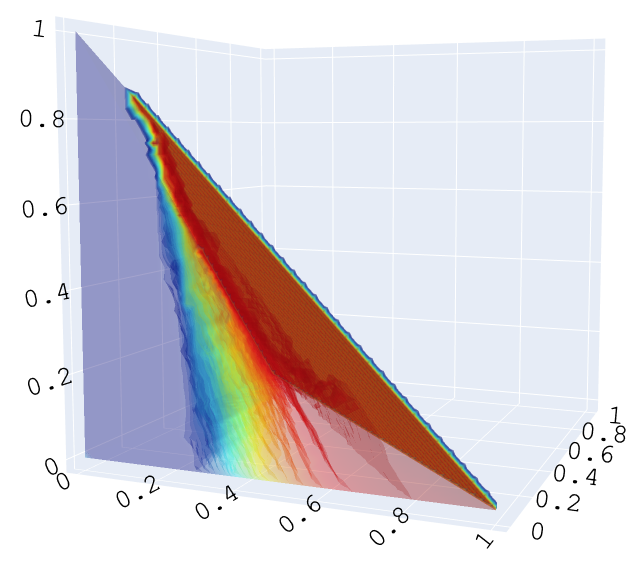}
\label{simp_spar_subfig2}}
\qquad
\subfloat[Subfigure 3 list of figures text][$\Fdfive{3}$ vs. $\wh \Delta_{[40]}$]{
\includegraphics[width=0.42\textwidth]{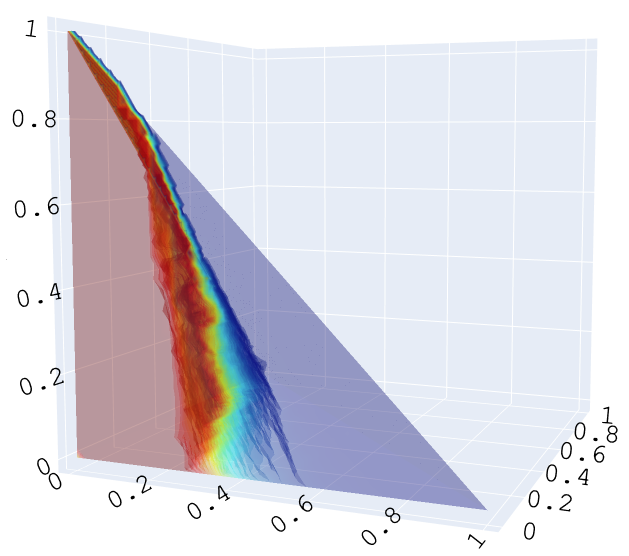}
\label{simp_spar_subfig3}}
\subfloat[Subfigure 4 list of figures text][$\Fdfive{4}$ vs. $\wh \Delta_{[40]}$]{
\includegraphics[width=0.42\textwidth]{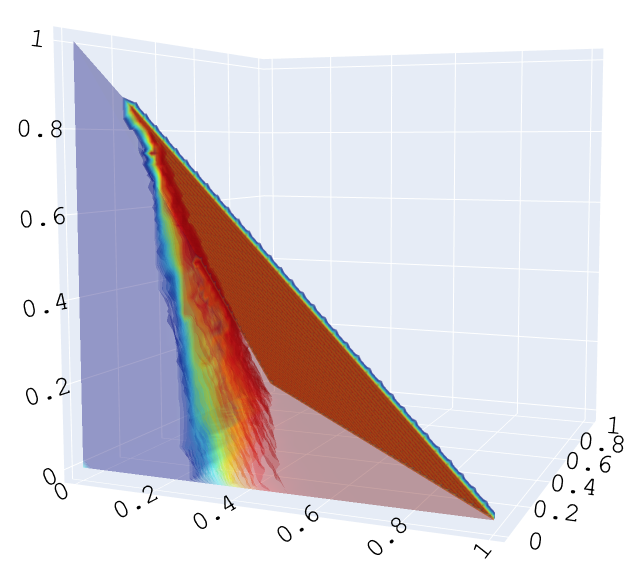}
\label{simp_spar_subfig4}}
\qquad
\centering
\includegraphics[width=0.5\textwidth]{images/colour_bar.png}
\caption{This figure shows the transparent three-dimensional plots of $\Fdfive{i}$ ($i \in [4]$) (in color) versus $(\lambda_2,\lambda_3,\lambda_4) \in \wh \Delta_{[40]}$, viewed from the facet: $\{(\lambda_2, \lambda_3, \lambda_4)  \in \wh \Delta:  \lambda_4 = 0 \}$.}
\label{spar_simp_fig}
\end{figure}

\begin{figure}[ht]
\centering 
    \subfloat[Subfigure 1 list of figures text][$f^{[40]}(w_\lambda)\in \Ffivestar{0.01}$ vs.\\ $(\lambda_2,\lambda_3,\lambda_4) \in \wh \Delta_{[40]}$ s.t. $w_\lambda \in \Wdkfives{0.01}$]{
	\includegraphics[scale=0.5]{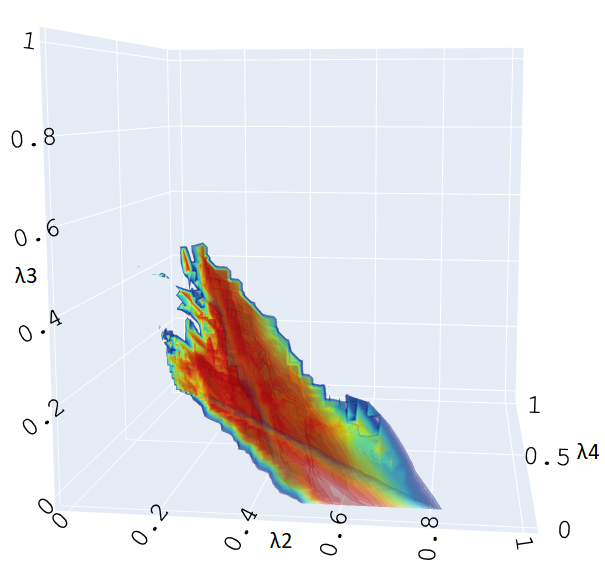}
	\label{spar_simp_superior_subfig1}}
	\subfloat[Subfigure 2 list of figures text][$\Ffivestar{0.01}$ vs. $\wh \Delta_{[40]} \cap \Lambda_{\Delta}$]{
	\includegraphics[scale=0.4]{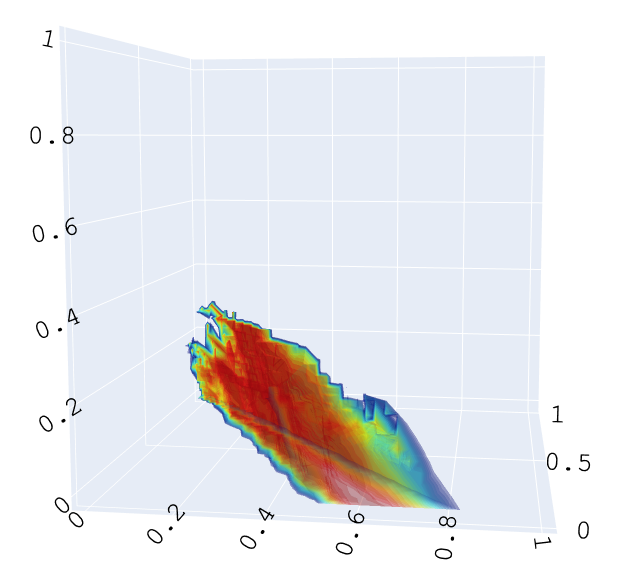}
	\label{spar_simp_superior_subfig2}}
	\qquad
	\centering
	\includegraphics[width=0.5\textwidth]{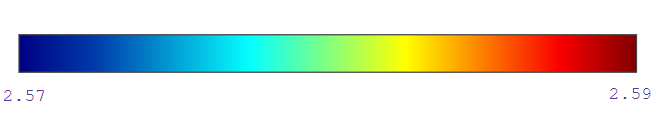}
	\caption{ This figure shows the transparent three-dimensional plot of $f^{[40]}(w_\lambda) \in \Ffivestar{0.01}$ in color versus $(\lambda_2,\lambda_3,\lambda_4) \in  \wh \Delta_{[40]}$ such that $w_\lambda \in \Wdkfives{0.01}$.} 
\label{F_Pareto_spar_simp_sum}
\end{figure}
\FloatBarrier
The similarity between \cref{spar_simp_fig} and its dense analog \cref{simp_fig} is because at least half of the portfolios $w_\lambda\in \Wdkfive$ are from $ \Wd$. Recall the histogram in \cref{supp_hist}, in which more than half of the points $ w_\lambda \in \Wd$ are shown to have support size five or less. Following the procedure of \cref{sub_sub_sec_spar_fista}, these portfolios are taken as they are when constructing $\Wdkfive$.

 Between  \cref{F_Pareto_spar_simp_sum} and \cref{F_Pareto_simp_sum}, there is again much similarity. The reader may wonder why the range of values in the sparse setting \cref{F_Pareto_spar_simp_sum} (from 2.57 to 2.59) exceeds that of dense setting \cref{F_Pareto_simp_sum} (from 2.455 to 2.475). There is no contradiction here because the scaling (\ref{scaling}) is different in each setting (simplex, cube, sparse, and dense). Hence, the values $\Ffivestar{0.01}$ and $\Fstar{0.01}$ are incomparable, similarly for the forthcoming cube setting.


\pagebreak
\subsubsection{Numerical results in the cube setting: $w \in [-1,1]^n$} \label{par_opt_cube}
\begin{figure}[!htbp]
\centering
\subfloat[Subfigure 1 list of figures text][ $\Fs{1}$ vs. $\wh \Delta_{[40]}$]{
\includegraphics[width=0.42\textwidth]{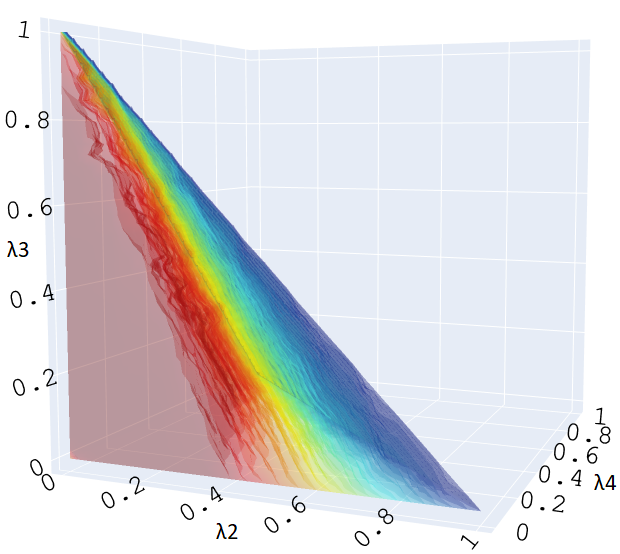}}
\subfloat[Subfigure 2 list of figures text][$\Fs{2}$ vs. $\wh \Delta_{[40]}$]{
\includegraphics[width=0.42\textwidth]{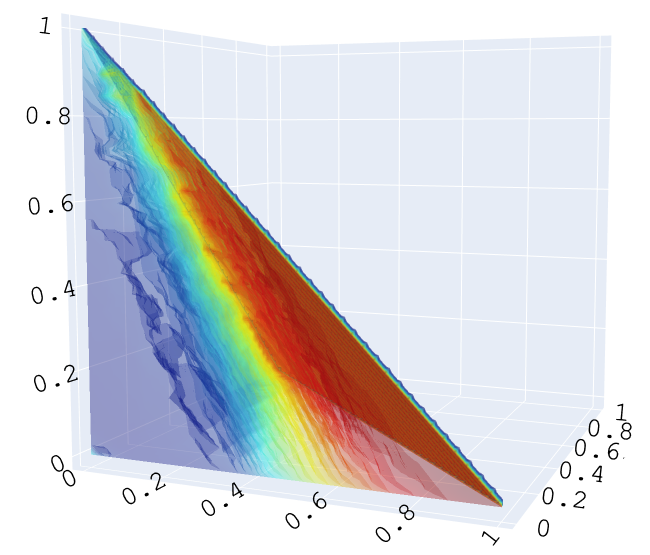}}
\qquad
\subfloat[Subfigure 3 list of figures text][$\Fs{3}$ vs. $\wh \Delta_{[40]}$]{
\includegraphics[width=0.42\textwidth]{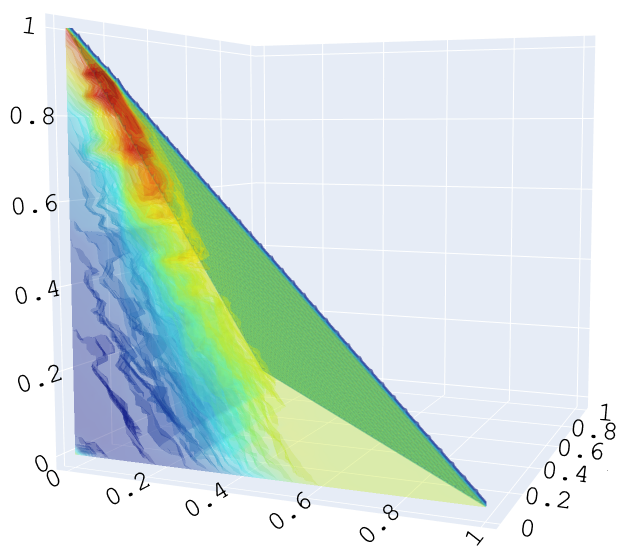}
\label{cube_skew}}
\subfloat[Subfigure 4 list of figures text][$\Fs{4}$ vs. $\wh \Delta_{[40]}$]{
\includegraphics[width=0.42\textwidth]{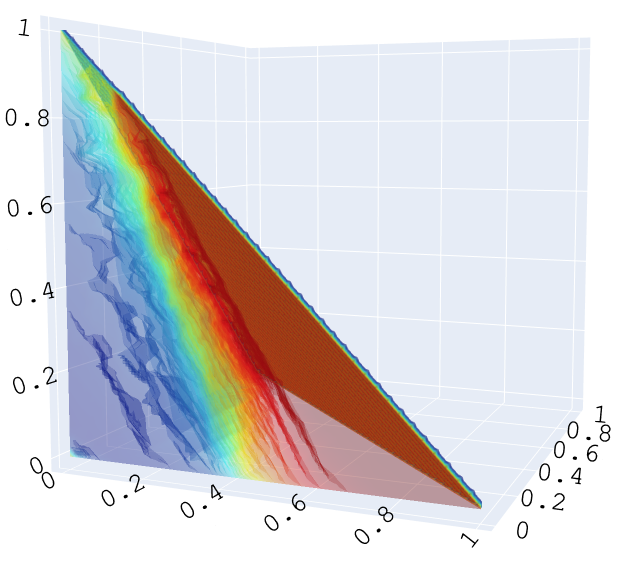}}
\qquad
\centering
\includegraphics[width=0.5\textwidth]{images/colour_bar.png}
\caption{This figure shows the transparent three-dimensional plots of $\Fis$ ($i \in [4]$) (in color) versus $(\lambda_2,\lambda_3,\lambda_4) \in \wh \Delta_{[40]}$, viewed from the facet: $\{(\lambda_2, \lambda_3, \lambda_4)  \in \wh \Delta:  \lambda_4 = 0 \}$.}
\label{cube_fig}
\end{figure}

\begin{figure}[!htbp]
\centering 
    \subfloat[Subfigure 1 list of figures text][$f^{[40]}(w_\lambda)\in \Fstar{0.025}$ vs. $(\lambda_2,\lambda_3,\lambda_4) \in \wh \Delta_{[40]}$ s.t. $w_\lambda \in \Wds{0.025}$]{
	\includegraphics[scale=0.53]{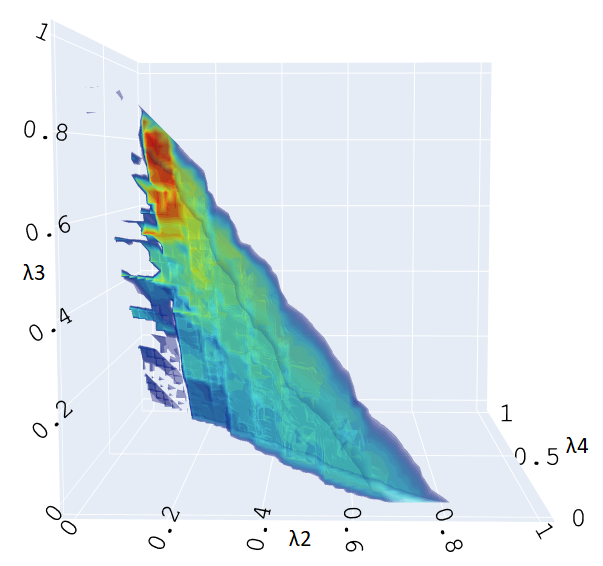}}
	\subfloat[Subfigure 2 list of figures text][$\Fstar{0.025}$ vs. $\wh \Delta_{[40]} \cap \Lambda_{\square}$]{
	\includegraphics[scale=0.4]{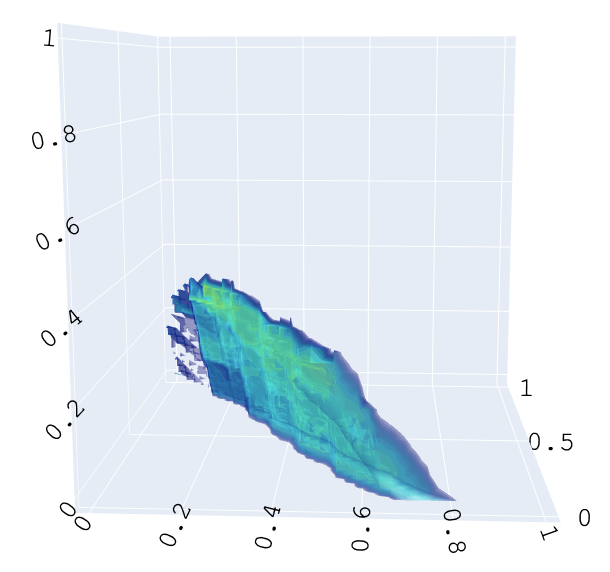}}
	\qquad
	\centering
	\includegraphics[width=0.5\textwidth]{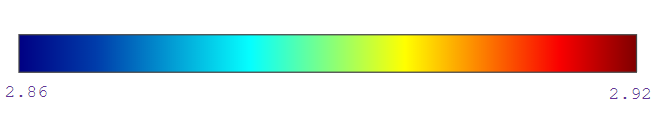}
	\caption{ This figure shows the transparent three-dimensional plot of $f^{[40]}(w_\lambda) \in \Fstar{0.025}$ in color versus $(\lambda_2,\lambda_3,\lambda_4) \in  \wh \Delta_{[40]}$ such that $w_\lambda \in \Wds{0.025}$.} 
\label{F_Pareto_cube_sum}
\end{figure}
\FloatBarrier
The cube setting differs significantly from the simplex setting. 
Portfolios $w_\lambda$ are now in $[-1,1]^n$ and have full support (at least for all examples we have computed).
Except for skewness, \cref{cube_skew}, the figures of \cref{cube_fig} follow roughly the same pattern as in \cref{simp_fig,spar_simp_fig}.
In the cube setting, portfolios $w_\lambda \in \Ws$ now require a large $\lambda_3$ to attain good values for the skewness objective, see \cref{cube_skew}. 

We observe that the portfolios of superior trade-off are more scarce in the cube setting than in the simplex counterpart.
Hence, in \cref{F_Pareto_cube_sum}, we now take $\eta=0.025$ because the set $\wh\Delta_\square^{[40],0.025}$ (of hyper-parameters corresponding to solutions of superior trade-off) gives a fuller and more informative plot than $\wh\Delta_\square^{[40],0.01}$. Secondly, we observe the same ``wedge" of superior portfolios we saw in \cref{F_Pareto_simp_sum,F_Pareto_spar_simp_sum}. Lastly, the portfolios $w_\lambda \in \Wss{0.025}$ that do the best in \cref{F_Pareto_cube_sum} have $\lambda_3 \geq 0.5$, with the concentration lying outside of $\Lambda_\square$.

\pagebreak
\subsubsection{Numerical results in the sparse cube setting: $w \in [-1,1]^n$, $|\supp(w)|\leq 5 $} \label{par_opt_cube_spar}

\begin{figure}[!htbp]
\centering
\subfloat[Subfigure 1 list of figures text][ $\Fsfive{1}$ vs. $\wh \Delta_{[40]}$]{
\includegraphics[width=0.42\textwidth]{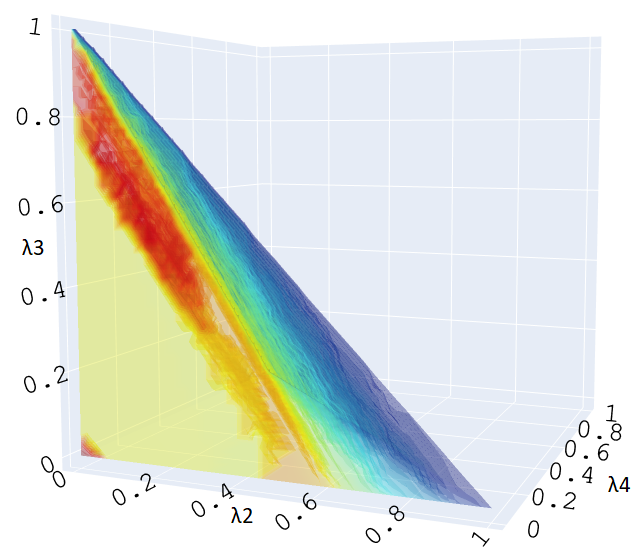}
\label{spar_cube_sub_1}}
\subfloat[Subfigure 2 list of figures text][$\Fsfive{2}$ vs. $\wh \Delta_{[40]}$]{
\includegraphics[width=0.42\textwidth]{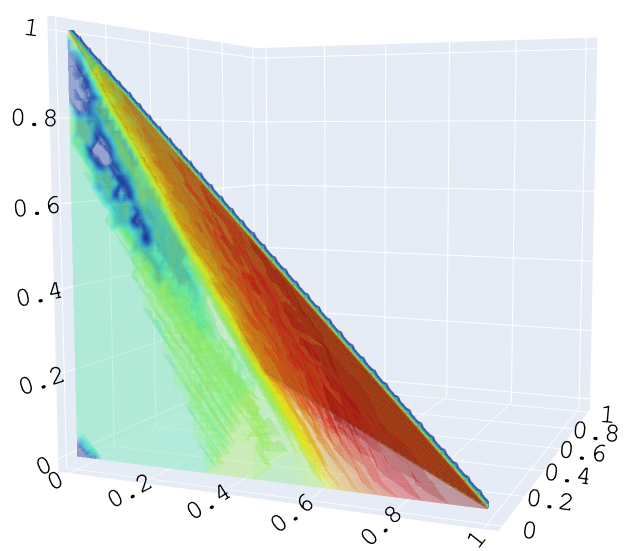}
\label{spar_cube_sub_2}}
\qquad
\subfloat[Subfigure 3 list of figures text][$\Fsfive{3}$ vs. $\wh \Delta_{[40]}$]{
\includegraphics[width=0.42\textwidth]{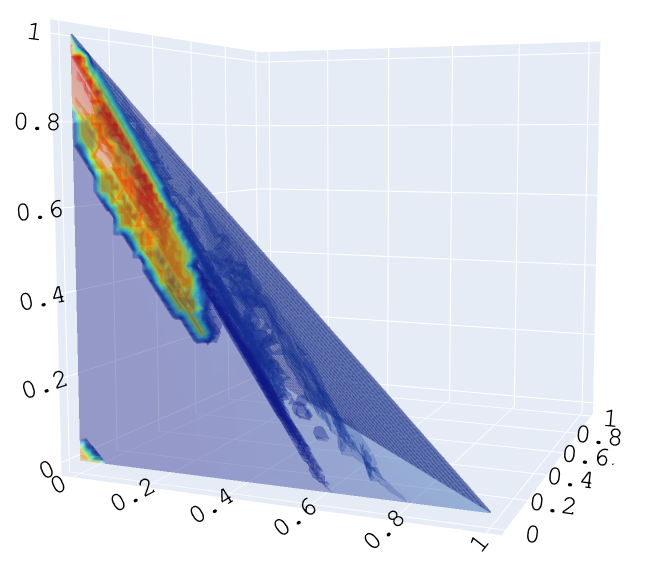}
\label{spar_cube_sub_3}}
\subfloat[Subfigure 4 list of figures text][$\Fsfive{4}$ vs. $\wh \Delta_{[40]}$]{
\includegraphics[width=0.42\textwidth]{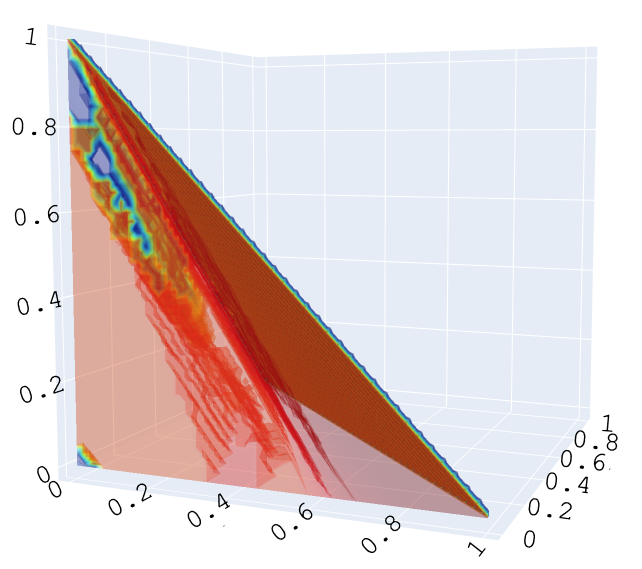}
\label{spar_cube_sub_4}}
\qquad
\centering
\includegraphics[width=0.5\textwidth]{images/colour_bar.png}
\caption{This figure shows the transparent three-dimensional plots of $\Fsfive{i}$ ($i \in [4]$) (in color) versus $(\lambda_2,\lambda_3,\lambda_4) \in \wh \Delta_{[40]}$, viewed from the facet: $\{(\lambda_2, \lambda_3, \lambda_4)  \in \wh \Delta:  \lambda_4 = 0 \}$.}
\label{sparse_cube_fig}
\end{figure}

\begin{figure}[!htbp]
\centering 
    \subfloat[Subfigure 1 list of figures text][$ \Fstar{0.025}$ vs. $\wh \Delta_{[40]}$]{
	\includegraphics[scale=0.5]{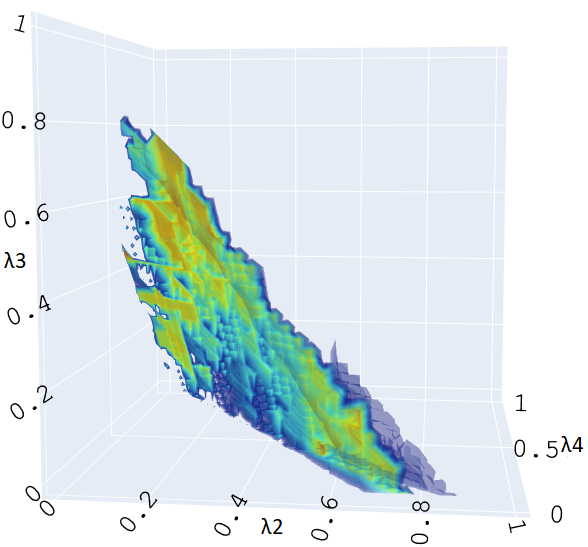}}
	\subfloat[Subfigure 2 list of figures text][$\Fstar{0.025}$ vs. $\wh \Delta_{[40]} \cap \Lambda_{\Delta}$]{
	\includegraphics[scale=0.40]{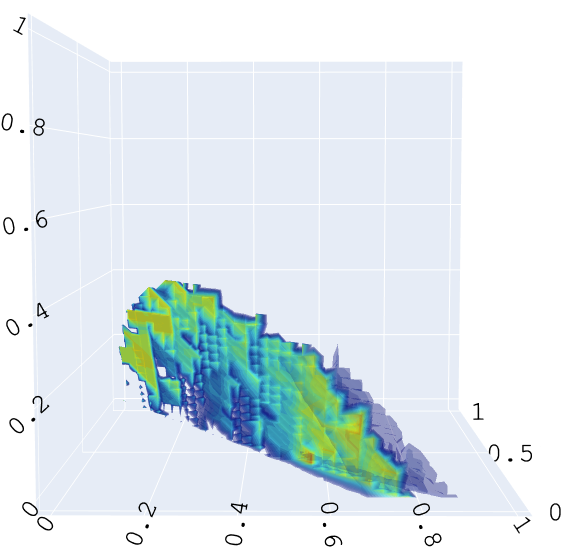}}
	\qquad
	\centering
	\includegraphics[width=0.5\textwidth]{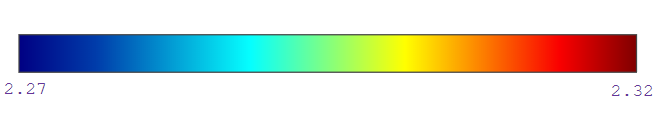}
	\caption{ This figure shows the transparent three-dimensional plot of $f^{[40]}(w_\lambda) \in \Ffivestars{0.025}$ in color versus $(\lambda_2,\lambda_3,\lambda_4) \in  \wh \Delta_{[40]}$ such that $w_\lambda \in \Wskfives{0.025}$.}
\label{F_Pareto_spar_cube_sum}
\end{figure}

\FloatBarrier

The results for the sparse cube setting differ vastly from the dense cube setting.
The difference in results is primarily due to the portfolios $w_\lambda \in \Ws$ having full support and thus differing greatly from the portfolios  $w_\lambda \in \Wskfive$.
In particular, we see concentrations forming in the same places as in \cref{cube_skew}, namely the upper tip of $\wh \Delta_{[40]}$ where $\lambda_3$ is large.
We also see a tinny concentration near $\lambda_1 = 1$. 
Despite the changes we still have that the odd objectives (mean and skewness) perform better in regions where the even objectives (variance and kurtosis) do poorly and vice versa, see \cref{sparse_cube_fig}.
Surprisingly, $\wh\Delta_{\square,5}^{[40],0.025}$ in \cref{F_Pareto_spar_cube_sum} is again the same ``wedge"-like shape we have seen in \cref{F_Pareto_simp_sum,F_Pareto_spar_simp_sum,F_Pareto_cube_sum}. There are now hardly any red regions, indicating that very few points reach the higher value range.

\section{Conclusion}\label{sec_conc}
In this paper, we considered the multi-objective optimization problem MVSK that models the portfolio selection problem in finance.
We considered two settings.
The first was the simplex ($w \in \Delta^n$), which represents portfolios that do not allow short selling and leverage.
The second setting was the cube ($w \in [-1,1]^n$), where we allow leverage and short selling.
Refer to \cref{sub_lev_short} to refresh the notions of leverage and short selling.  
Furthermore, we introduced a sparse variant of MVSK, where one can set an upper bound on the support size of portfolios.\\

In order to (partially) recover the Pareto front of MVSK (for the different settings), we proposed the following three-step process.
First, we considered the linear scalarization of MVSK and identified a set of hyper-parameters for which the resulting scalar-objective problems are convex.
Second, we used the fast iterative shrinkage-thresholding algorithm (FISTA) (which converges to a global optimizer when applied to convex problems) to recover optimizers of the scalarized problems.
Third, we used the fact that the (global) optimizers of neat scalarizations are Pareto optimizers of the original multi-objective problem, hence, showing that many of the optimizers computed by FISTA are Pareto optimal for MVSK.
 
Additionally, we demonstrate that gradient-descent algorithms like FISTA have three desirable properties in this setting: the ability to benefit from warm starts, the tendency to generate sparse solutions (due to the projection step, in the case of the simplex), and guaranteed convergence to the global optimizer (for convex problems).
We hope the visualizations accompanying our computations will further intuition and understanding within this exciting field.\\ 

In the introduction, we hinted at extending the model to include even higher moments, i.e., beyond kurtosis.
The formulation of the objective functions in (\ref{mom_funs_k}) is well-defined for any integer $k>4$.
Hence, one can define objectives $f_k$ with $k>4$ in addition to those already present in (\ref{MVSK}) and thereby extend the model. With an extended model, one can again scalarize linearly, now using more hyper-parameters than before. Again one can characterize the Hessian of this new scalarization and possibly recover results similar to \cref{pos_coef,cor_convex}. 
However, one must first justify adding these higher moments, considering the additional computational and complexity costs involved.
There is not much motivation in the literature for this further extension.
Some authors even advocate against relying on correlation-based risk measures \cite{taleb2020statistical}.\\

\noindent Alternative to the linear scalarization (\ref{F_lambda}) some authors propose the Minkowski distance scalarization (\ref{minkow_scal}). The Minkowski distance formulation lends itself to a signomial optimization interpretation \cite{Chandrasekaran}.
Indeed, the scalarization (\ref{minkow_scal}) applied to the MVSK problem with simplex domain has the following form:
$$
\min_{w\in \Delta^n} \sum_{k \in [4]}\big| f_k(w) -f_k^* \big|^{\lambda_k} ,
$$
where $f_k^* :=  \min_{w \in \Delta} (-1)^{k} f_k(w)$ with $f_k$ given in (\ref{mom_funs}) and (\ref{mean_fun}) for each $k \in [4]$. Under the change of variable $\exp(u) :=  (\exp(u_i))_{i \in [n]} := (w_i)_{i \in [n]}$, the above problem can be written as a signomial optimization program
\begin{equation*}  
\begin{split} 
\min & \sum_{k \in [4]} \exp(\lambda_kv_k)\\
\mathrm{s.t.} & \sum_{i \in [n]} \exp(u_i) = 1 \\
& \exp(v_k) =  (-1)^{k}f_k(\exp(u)) + (-1)^{k+1}f_k^*    ~(k \in [4]).
\end{split}
\end{equation*} 
Problems of this type have been studied before and have mature methods to solve them, see \cite{murraySignomialPolynomialOptimization2021}.
Approaching the MVSK problem from the signomial programming direction opens a new and unexplored line of inquiry into the MVSK problem.
As before, one can try to characterize the convexity of such a scalarization in the hopes of getting similar results to \cref{pos_coef}, but we do not attempt it here.
The appeal of solving these scalarizations is that they are also neat for $\lambda >0$ and could reveal Pareto points that the linear scalarization approach cannot. However, the benefits of the Minkowski distance scalarization must be weighed against the fact that it is much harder to interpret than linear scalarization. Moreover, one has to compute the independent optima $f_k^*$ for $k\in [4]$, which can already be challenging in the case of $k=3$.

\bibliographystyle{abbrv} 
\bibliography{MVSK.bib}

\end{document}